\documentclass[aps,prl,reprint,superscriptaddress,amsmath,amssymb,longbibliography,floatfix]{revtex4-2}
\usepackage[T1]{fontenc}
\usepackage[utf8]{inputenc}
\usepackage{lmodern}
\usepackage{amsthm}
\usepackage{etoolbox}
\usepackage{microtype}
\usepackage{graphicx}
\usepackage{xcolor}
\usepackage{tikz}
\usetikzlibrary{arrows.meta,decorations.pathreplacing,decorations.markings}
\usepackage[colorlinks=true,linkcolor=blue,citecolor=blue,urlcolor=blue,hypertexnames=false]{hyperref}
\hypersetup{pdftitle={Square-root growth of operator entanglement in an integrable brickwork circuit},
 pdfauthor={Balazs Pozsgay},
 pdfsubject={PRL submission with Supplemental Material}}
\newcommand{\ket}[1]{|#1\rangle}
\newcommand{\bra}[1]{\langle#1|}
\newcommand{\Ftwo}{\mathbb F_2}
\newcommand{\id}{\mathbf 1}
\newcommand{\op}{\mathrm{op}}
\newcommand{\dw}{\mathrm{dw}}
\newcommand{\Bin}{\operatorname{Bin}}
\newcommand{\OSR}{\operatorname{OSR}}

\newcommand{\E}{\mathbb E}
\newcommand{\Prob}{\mathbb P}
\newcommand{\HS}{\mathrm{HS}}
\newcommand{\rank}{\operatorname{rank}}
\newcommand{\head}[1]{\textit{#1.---}}

\newtheorem*{unnumberedlemma}{Lemma}
\makeatletter
\let\arxiv@affiliation\affiliation
\def\arxiv@part{main}
\patchcmd{\NAT@bibsetnum}{\ref{LastBibItem}}%
  {\ref{\arxiv@part-LastBibItem}}{}{\errmessage{Cannot patch bibliography width}}
\patchcmd{\endrtx@thebibliography}{\label{LastBibItem}}%
  {\label{\arxiv@part-LastBibItem}}{}{\errmessage{Cannot patch bibliography label}}
\let\endthebibliography\endrtx@thebibliography
\makeatother
\begin{document}
\title{Square-root growth of operator entanglement in an integrable brickwork circuit}
\author{Bal\'azs Pozsgay}
\affiliation{MTA-ELTE ``Momentum'' Integrable Quantum Dynamics Research Group,\protect\\
ELTE E\"otv\"os Lor\'and University, Budapest, Hungary}
\date{\today}
\begin{abstract}
It is widely expected that the von Neumann entanglement of a local operator
grows at most logarithmically in integrable many-body systems in infinite volume. We give a counterexample in a
four-state brickwork circuit whose gate is a permutation matrix solving
the constant Yang--Baxter equation. The von Neumann operator entropy grows as
$(\log2)\sqrt{t/\pi}+O(\log t)$, whereas fixed-index R\'enyi entropies grow
linearly below index one and logarithmically above it.
At fixed relative Hilbert--Schmidt error below one, matrix product operator
simulations require a bond dimension growing at least as $\exp(c\sqrt t)$
for some $c>0$.
\end{abstract}
\maketitle

\head{Introduction}
Operator entanglement measures the complexity developed by a local operator
under Heisenberg evolution. Integrable models are special many-body systems
because they possess an extensive family of conserved quantities that
constrain the dynamics. While the real-space entanglement of pure quantum
states follows the same overall scaling laws in integrable and chaotic
systems, various examples have shown that the entanglement of local operators
behaves very differently. In particular, several examples were found where operator entanglement in integrable systems
grows only logarithmically or it even remains bounded
\cite{main-ProsenPizorn2007,main-PizornProsen2009,main-Dubail2017,main-KlobasMedenjakProsenVanicat2019,main-AlbaDubailMedenjak2019,main-Medenjak2022,main-MuthUnanyanFleischhauer2011,main-MurcianoDubailCalabrese2024,main-Wang2025}.
These results motivated the broader expectation that the von Neumann
entanglement of a local operator grows at most logarithmically under
integrable dynamics in infinite volume~\cite{main-ProsenZnidaric2007,main-BertiniKosProsenI,main-Alba2021,main-JacobyGopalakrishnan2026,main-Dowling2026}.

Against this background, we ask whether integrability itself enforces the
expected logarithmic bound. The answer is no. We construct a brickwork
circuit with local dimension four and a one-site operator whose von Neumann
entropy grows as $(\log2)\sqrt{t/\pi}+O(\log t)$. The gate is a unitary solution of the braid equation (we
say it is a Yang--Baxter gate) and it is also involutive. 
It follows that the model is Yang--Baxter integrable.
The circuit also supports an exponentially large family of ballistic gliders. Thus neither structure
by itself ensures logarithmic operator-entanglement growth. For comparison, square-root entanglement laws have been found in distinct settings: colored Motzkin and Fredkin ground states~\cite{main-MovassaghShor2016,main-SalbergerKorepin2017,main-SuginoKorepin2018}, stochastic state dynamics~\cite{main-CaoTilloyDeLuca2019,main-Znidaric2020,main-Alba2026QSSEP}, and operator entanglement under stochastic free-fermion dynamics~\cite{main-TonettoDubail2026}. To our knowledge, however, this square-root operator-entanglement scaling had not previously been found in a deterministic integrable circuit.

Known logarithmic examples in this setting relied on additional mechanisms, including
free-fermion structure \cite{main-ProsenPizorn2007,main-PizornProsen2009}, exceptionally
simple scattering \cite{main-AlbaDubailMedenjak2019}, or a single-file rule for
internal degrees of freedom \cite{main-Medenjak2022}. A recent study of
several Yang--Baxter-gate families identified further sufficient conditions
for bounded or logarithmic operator entanglement \cite{main-sajat-opent}.
As a special case it was proven that involutive and dual-unitary~\cite{main-BertiniClaeysProsen2026} Yang--Baxter gates always lead to logarithmic growth.
The present gate evades these conditions: it is involutive but not dual-unitary.

We will also treat a simpler Schr\"odinger-picture problem, a quench from
a product-state domain wall.
This simpler computation displays all the technical steps needed for the operator
entanglement. 

The counterexample was found by ChatGPT 6 Astra, building on work with the
earlier model 5.6 Sol. All computations were checked by the author. The
manuscript was prepared in collaboration between the human author and the AI.

\head{Circuit and operator entanglement}
Consider an infinite chain with two-site gate $R$ and one Floquet period $U_F=U_{\rm odd}U_{\rm even}$, where
\begin{equation}
 U_{\rm even}=\prod_j R_{2j,2j+1},\quad
 U_{\rm odd}=\prod_j R_{2j-1,2j}.
 \label{main-eq:floquet}
\end{equation}
Time $t$ counts complete periods. Throughout, we take the system size
$L\to\infty$ at fixed $t$ before taking the long-time limit $t\to\infty$.
We place a spatial cut between sites 0 and 1.
Figure~\ref{main-fig:circuit-gate} shows two periods in braid notation.

For a one-site operator $O$ at site $1$, its Heisenberg time evolution is given by
$O(t)=U_F^tOU_F^{-t}$.\footnote{This definition of Heisenberg time evolution differs from the standard convention, but we adopt it to maintain a full mathematical parallel between the state and operator problems.}
We define its operator space entanglement as follows \cite{main-Zanardi2001}.

We first normalize the operator in the Hilbert--Schmidt norm and then
write its operator-Schmidt decomposition as
\begin{equation}
 \frac{O(t)}{\|O(t)\|_{\HS}}
 =\sum_j\sqrt{p_j(t)}\,A_j(t)\otimes B_j(t),
 \label{main-eq:schmidt}
\end{equation}
where $\operatorname{tr}(A_i^\dagger A_j)
=\operatorname{tr}(B_i^\dagger B_j)=\delta_{ij}$ and $\sum_jp_j=1$.
The number of nonzero terms is the exact bond dimension needed across this
cut in a matrix-product-operator representation. For $0<\alpha\ne1$,
the R\'enyi entropies and their von Neumann limit are defined as
\begin{equation}
 S_\alpha^\op=\frac{\log\sum_jp_j^\alpha}{1-\alpha},\qquad
 S_1^\op=-\sum_jp_j\log p_j.
 \label{main-eq:entropies}
\end{equation}
We use natural logarithms.

\head{The two-site gate}
The local space is $H_{\rm sec}\otimes H_{\rm col}$, with
$H_{\rm sec}\simeq H_{\rm col}\simeq\mathbb C^2$.
Write a basis state as $\ket{s,a}$. Here we call $s$ the {\it sector bit} and we will use the notation $s=A,B$.
The second bit $a\in\Ftwo$ is interpreted as a color. The dynamics for the sector and color bits will be markedly different.

Our gate is a permutation matrix, and it is obtained as the linear extension of a classical map $r: X^2\to X^2$ where
$X=\Ftwo\times \Ftwo$, such that
$R\ket{s,a;\tau,b}=\ket{r(s,a;\tau,b)}$.
We define the classical map as
\begin{equation}
r(s,a;\tau,b)
 =(\tau,a+(s+\tau)b;s,b),
 \label{main-eq:gate}
\end{equation}
where binary arithmetic is understood. If $s=\tau$ then the map is equal to the identity. If $s\ne \tau$ then the sector
labels are exchanged, the color $b$ is 
kept invariant, and the color $a$ suffers a shear. If we interpret the gate as describing diagonal scattering of particles with
two labels (see the right-hand panel of Fig.~\ref{main-fig:circuit-gate}), then we can observe that the sector labels $s$ and $\tau$ propagate ballistically. The invariance of the color $b$ is then
interpreted as a perfect reflection.

For unequal sectors, write the input colors as $(a,b)$ and the output
colors as $(h(a,b),b)$, where $h(a,b)=a+b$. This color rule has
one-sided non-degeneracy: knowing $a$ and $h(a,b)$ fixes $b$, whereas
knowing the two right-hand colors, both equal to $b$, does not fix $a$.
This permits local reconstruction analogous to that of dual-unitary
gates~\cite{main-GomborPozsgay2022}, but reconstruction in the opposite
diagonal direction fails. We will use this property in the Supplemental
Material.

\begin{figure}[!tb]
\centering
\begin{tikzpicture}[x=1cm,y=1cm,line cap=round,line join=round,
 font=\footnotesize,
 wire/.style={draw=black!78,line width=.6pt},
 vertex/.style={circle,draw=blue!65!black,fill=white,
   minimum size=1.8mm,inner sep=0pt,line width=.65pt}]
 \begin{scope}[x=.516cm,y=.3752724cm,xshift=-.25cm]
  \foreach \x/\xp in {0/1,1/0,2/3,3/2,4/5,5/4,6/7,7/6}
    \draw[wire] (\x,.35)--(\xp,1.25);
  \foreach \x in {.5,2.5,4.5,6.5}
    \node[vertex] at (\x,.80) {};
  \foreach \x in {0,7}\draw[wire] (\x,1.25)--(\x,2.15);
  \foreach \x/\xp in {1/2,2/1,3/4,4/3,5/6,6/5}
    \draw[wire] (\x,1.25)--(\xp,2.15);
  \foreach \x in {1.5,3.5,5.5}
    \node[vertex] at (\x,1.70) {};
  \foreach \x/\xp in {0/1,1/0,2/3,3/2,4/5,5/4,6/7,7/6}
    \draw[wire] (\x,2.15)--(\xp,3.05);
  \foreach \x in {.5,2.5,4.5,6.5}
    \node[vertex] at (\x,2.60) {};
  \foreach \x in {0,7}\draw[wire] (\x,3.05)--(\x,3.95);
  \foreach \x/\xp in {1/2,2/1,3/4,4/3,5/6,6/5}
    \draw[wire] (\x,3.05)--(\xp,3.95);
  \foreach \x in {1.5,3.5,5.5}
    \node[vertex] at (\x,3.50) {};
  \foreach \x in {0,...,7}\draw[wire] (\x,3.95)--(\x,4.45);
 \end{scope}
 \begin{scope}[xshift=.3cm]
 \draw[wire] (5.45,.375)--(6.35,1.125);
 \draw[wire] (6.35,.375)--(5.45,1.125);
 \node[vertex] at (5.9,.75) {};
 \node[anchor=north,inner sep=1pt] at (5.45,.295) {$(s,a)$};
 \node[anchor=north,inner sep=1pt] at (6.35,.295) {$(\tau,b)$};
 \node[anchor=south east,inner sep=1pt] at (5.75,1.205) {$(\tau,a+(s+\tau)b)$};
 \node[anchor=south,inner sep=1pt] at (6.35,1.205) {$(s,b)$};
 \end{scope}
\end{tikzpicture}
\caption{Braid representation of the circuit and its local gate.
Left: two Floquet periods, with even and odd layers alternating from
bottom to top. Right: input and output sector and color variables for
the map~\eqref{main-eq:gate}. Hollow vertices denote $R$; time runs upwards.}
\label{main-fig:circuit-gate}
\end{figure}
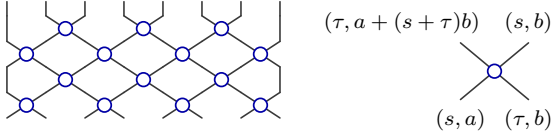

The gate is real, unitary, involutive,
and it satisfies the braid relation:
\begin{equation}
 R_{12}R_{23}R_{12}
 =R_{23}R_{12}R_{23}.
 \label{main-eq:braid}
\end{equation}
The gate is related to the twisted unions and gluing constructions studied
previously~\cite{main-MajidMarkl1996,main-GatevaIvanovaMajid2008,main-Nichita2013,main-GomborPozsgay2022}.

There are two important consequences of the braid relation and
involutivity. First, these relations define a unitary representation
of the symmetric group $S_n$ on the $n$-site Hilbert space, with each adjacent
transposition represented by $R_{j,j+1}$. We will use this natural
representation repeatedly, and for a permutation $p\in S_n$ we will denote the corresponding unitary operator as $U_p$. 

Second, involutivity permits the Baxterization
$\check{\mathcal R}(u)=(\id+iuR)/(1+iu)$, leading to a spectral-parameter-dependent solution of the Yang--Baxter
equation. 
The transfer-matrix construction of Ref.~\cite{main-sajat-opent}
provides commuting extensive charges with local densities, as in the
usual Yang--Baxter framework \cite{main-VanicatZadnikProsen2018}.
The braid relation also generates nontrivial ballistic gliders and an
exponentially large conserved family \cite{main-GomborPozsgay2022,main-sajat-opent};
in this sense the circuit is superintegrable.

Each gate preserves the numbers of $A$ and $B$ sector labels, giving a
$U(1)$ symmetry. We denote the number of $A$ labels by $N_A$.

The gate is not dual-unitary because it acts as the identity when the
sector labels coincide.

\head{A domain-wall quench}
To expose the entanglement mechanism we consider a simpler problem: real-time evolution of a selected pure state. This
state will be constructed as a domain-wall state such that away from the interface it is invariant under the circuit
dynamics.

Define $\ket{+}_{\rm sec}=(\ket{A}+\ket{B})/\sqrt{2}$ in sector space, and $\ket{+}_{\rm col}=(\ket{0}+\ket{1})/\sqrt2$ in color space.
Furthermore, define
\begin{equation}
 \ket{\psi_{\mathrm L}}=\ket{+}_{\rm sec}\otimes \ket{0}_{\rm col}\qquad
 \ket{\psi_{\mathrm R}}=\ket{+}_{\rm sec}\otimes \ket{+}_{\rm col}
 \label{main-eq:reservoirs}
\end{equation}
The initial state and its Schr\"odinger evolution are
\begin{equation}
 \ket{\Psi_0}=\bigotimes_{x\le0}\ket{\psi_{\mathrm L}}_x\otimes
 \bigotimes_{x\ge1}\ket{\psi_{\mathrm R}}_x,\qquad
 \ket{\Psi_t}=U_F^t\ket{\Psi_0}.
 \label{main-eq:quench}
\end{equation}
The state is unentangled across every spatial cut at $t=0$.
The domain wall is entirely in the color degree of freedom, since the
left and right states have the same sector component.
The circuit entangles the sector and color parts of the wave function,
so their initial factorization is lost.

Let us first check how the gate acts within either reservoir. We write
$\ket{s,\chi}=\ket{s}_{\rm sec}\otimes\ket{\chi}_{\rm col}$.
For $s,\tau\in\{A,B\}$ and $\chi=0,+$, Eq.~\eqref{main-eq:gate} gives
$R\ket{s,\chi;\tau,\chi}=\ket{\tau,\chi;s,\chi}$.
The shear leaves two zero colors unchanged, while its permutation action
preserves the uniform color superposition. By linearity, this SWAP action
also holds for arbitrary sector superpositions: crossings reorder the
sector states without changing the colors. Taking every sector state to
be $\ket{+}_{\rm sec}$, we obtain the reservoir invariance:
\begin{equation}
 R\ket{\psi_{\mathrm L},\psi_{\mathrm L}}=\ket{\psi_{\mathrm L},\psi_{\mathrm L}},\qquad
 R\ket{\psi_{\mathrm R},\psi_{\mathrm R}}=\ket{\psi_{\mathrm R},\psi_{\mathrm R}}.
 \label{main-eq:bulk-fixed}
\end{equation}

For each $t\ge1$, we now introduce an auxiliary circuit that generates
the same entanglement across the interface in this quench.
Let $w_t\in S_{2t}$ be the permutation that swaps two packets of length $t$.
We define $W_t=U_{w_t}$; see the left-hand panel of Fig.~\ref{main-fig:rectangle} for $t=3$.
The circuit consists of the braiding steps that arise when we exchange two packets of length $t$.

We define an evolved auxiliary state by
\begin{equation}
 \ket{\Phi_t}=W_t\bigl(\ket{\psi_{\mathrm L}}^{\otimes t}\otimes \ket{\psi_{\mathrm R}}^{\otimes t}\bigr).
 \label{main-eq:dw-rectangle}
\end{equation}
Its Schmidt spectrum equals that of
$\ket{\Psi_t}$ across $0|1$.
We prove this in two steps. We first remove gates outside the forward light cone of the initial
interface using~\eqref{main-eq:bulk-fixed}.
We then disregard gates whose future light cones lie entirely on one side
of the cut at time $t$, since they contribute only separate left and right
unitaries and preserve the Schmidt spectrum. The remaining network is $W_t$.
This is the state counterpart of the
reduction in Ref.~\cite{main-sajat-opent}.

\begin{figure}[!tb]
\centering
\begin{tikzpicture}[x=0.55cm,y=0.40cm,line cap=round,line join=round,
 font=\scriptsize,
 statewire/.style={draw=black!78,line width=.6pt},
 foldedwire/.style={draw=blue!65!black,line width=.75pt},
 statevertex/.style={circle,draw=blue!65!black,fill=white,
   minimum size=1.65mm,inner sep=0pt,line width=.65pt},
 foldedvertex/.style={circle,draw=blue!65!black,fill=blue!65!black,
   minimum size=1.65mm,inner sep=0pt},
 stateinput/.style={rectangle,draw=black!65,fill=black!4,
   minimum size=2.8mm,inner sep=.4pt},
 identity/.style={rectangle,draw=blue!65!black,fill=white,
   minimum size=3.4mm,inner sep=0pt,line width=.75pt},
 source/.style={rectangle,draw=orange!80!black,fill=orange!15,
   rounded corners=.7pt,minimum width=4.8mm,minimum height=3.2mm,
   inner sep=.2pt,line width=.75pt}]

 \begin{scope}
  \node[font=\small,anchor=west] at (-.45,6.45) {(a)};
  \foreach \x in {0,...,5}\draw[statewire](\x,-.35)--(\x,.35);
  \foreach \x in {0,1,4,5}\draw[statewire](\x,.35)--(\x,1.25);
  \foreach \x/\xp in {2/3,3/2}\draw[statewire](\x,.35)--(\xp,1.25);
  \node[statevertex] at (2.5,.80){};
  \foreach \x in {0,5}\draw[statewire](\x,1.25)--(\x,2.15);
  \foreach \x/\xp in {1/2,2/1,3/4,4/3}\draw[statewire](\x,1.25)--(\xp,2.15);
  \foreach \x in {1.5,3.5}\node[statevertex]at(\x,1.70){};
  \foreach \x/\xp in {0/1,1/0,2/3,3/2,4/5,5/4}\draw[statewire](\x,2.15)--(\xp,3.05);
  \foreach \x in {.5,2.5,4.5}\node[statevertex]at(\x,2.60){};
  \foreach \x in {0,5}\draw[statewire](\x,3.05)--(\x,3.95);
  \foreach \x/\xp in {1/2,2/1,3/4,4/3}\draw[statewire](\x,3.05)--(\xp,3.95);
  \foreach \x in {1.5,3.5}\node[statevertex]at(\x,3.50){};
  \foreach \x in {0,1,4,5}\draw[statewire](\x,3.95)--(\x,4.85);
  \foreach \x/\xp in {2/3,3/2}\draw[statewire](\x,3.95)--(\xp,4.85);
  \node[statevertex]at(2.5,4.40){};
  \foreach \x in {0,...,5}\draw[statewire](\x,4.85)--(\x,5.30);
  \foreach \x in {0,1,2}\node[stateinput]at(\x,-.35){$\psi_{\mathrm L}$};
  \foreach \x in {3,4,5}\node[stateinput]at(\x,-.35){$\psi_{\mathrm R}$};
  \node at(1,5.75){$\mathrm R_{\rm out}$};
  \node at(4,5.75){$\mathrm L_{\rm out}$};
  \draw[dashed,black!45](2.5,4.95)--(2.5,6.05);
 \end{scope}

 \begin{scope}[shift={(7.3,0)}]
  \node[font=\small,anchor=west] at (-.45,6.45) {(b)};
  \foreach \x in {0,...,5}\draw[foldedwire](\x,-.35)--(\x,.35);
  \foreach \x in {0,1,4,5}\draw[foldedwire](\x,.35)--(\x,1.25);
  \foreach \x/\xp in {2/3,3/2}\draw[foldedwire](\x,.35)--(\xp,1.25);
  \node[foldedvertex] at (2.5,.80){};
  \foreach \x in {0,5}\draw[foldedwire](\x,1.25)--(\x,2.15);
  \foreach \x/\xp in {1/2,2/1,3/4,4/3}\draw[foldedwire](\x,1.25)--(\xp,2.15);
  \foreach \x in {1.5,3.5}\node[foldedvertex]at(\x,1.70){};
  \foreach \x/\xp in {0/1,1/0,2/3,3/2,4/5,5/4}\draw[foldedwire](\x,2.15)--(\xp,3.05);
  \foreach \x in {.5,2.5,4.5}\node[foldedvertex]at(\x,2.60){};
  \foreach \x in {0,5}\draw[foldedwire](\x,3.05)--(\x,3.95);
  \foreach \x/\xp in {1/2,2/1,3/4,4/3}\draw[foldedwire](\x,3.05)--(\xp,3.95);
  \foreach \x in {1.5,3.5}\node[foldedvertex]at(\x,3.50){};
  \foreach \x in {0,1,4,5}\draw[foldedwire](\x,3.95)--(\x,4.85);
  \foreach \x/\xp in {2/3,3/2}\draw[foldedwire](\x,3.95)--(\xp,4.85);
  \node[foldedvertex]at(2.5,4.40){};
  \foreach \x in {0,...,5}\draw[foldedwire](\x,4.85)--(\x,5.30);
  \foreach \x in {0,1,2,4,5}\node[identity]at(\x,-.35){$\iota$};
  \node[source]at(3,-.35){$x_O$};
  \node at(1,5.75){$\mathrm R_{\rm out}$};
  \node at(4,5.75){$\mathrm L_{\rm out}$};
  \draw[dashed,black!45](2.5,4.95)--(2.5,6.05);
 \end{scope}
\end{tikzpicture}
\caption{Entangling rectangles for $t=3$.
(a) The state problem, with incoming packets prepared in
$\psi_{\mathrm L}$ and $\psi_{\mathrm R}$.
(b) The folded operator problem, with identity inputs $\iota$ and the local
source $x_O$. Hollow and filled vertices denote $R$ and $\widehat R$,
respectively. In both panels the dashed line marks the outgoing packet cut.}
\label{main-fig:rectangle}
\end{figure}
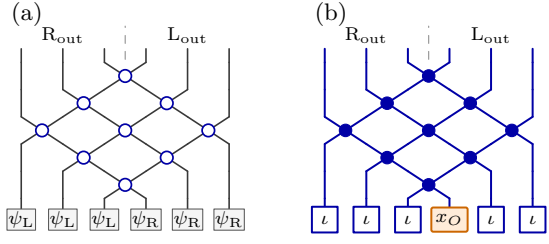

\head{The entanglement} We resolve the incoming state into sectors with
fixed numbers of $A$ labels in each packet.
Let $\Pi_{t,k}$ be the projector on $t$ sites to the subspace with exactly $k$ $A$ labels. Then we expand the state as
\begin{equation}
 \ket{\Phi_t}=\sum_{k,\ell}\sqrt{\pi^{\dw}_{k\ell}}\ket{\phi_{k\ell}},
 \qquad \pi^{\dw}_{k\ell}=4^{-t}\binom tk\binom t\ell,
 \label{main-eq:dw-branches}
\end{equation}
where $\ket{\phi_{k\ell}}$ is a normalized vector defined via
\begin{equation}
  \sqrt{\pi^{\dw}_{k\ell}}\ket{\phi_{k\ell}}=W_t   \bigl[
  \bigl(\Pi_{t,\ell}\ket{\psi_{\mathrm L}}^{\otimes t}\bigr)
  \otimes \bigl(\Pi_{t,k}  \ket{\psi_{\mathrm R}}^{\otimes t} \bigr)\bigr]
 \label{main-eq:count-transport}
\end{equation}
The deterministic transport of the sector labels through $W_t$ implies that $\ket{\phi_{k\ell}}$ contains $k$ and $\ell$
$A$ labels on the left and right, respectively.
Therefore, formula \eqref{main-eq:count-transport} is a symmetry-resolved expression of the final state, but it is not yet a Schmidt
decomposition. Our main strategy is to compute the 
entanglement for each $\ket{\phi_{k\ell}}$, and then to use the double sum to give lower and upper bounds for
the entanglement in $\ket{\Phi_t}$. 

We expand the initial product states on the left and on the right according to their sector labels, while keeping the color
part of the wave function fixed.
Direct expansion gives the states
\begin{equation}
\ket{\Phi_t((s_j)|(\tau_j))}\equiv W_t\left[  \bigl(\otimes_{j} \ket{s_j,0}\bigr)\otimes
\bigl(    \otimes_{i} \ket{\tau_i,+}\bigr)\right]
\end{equation}
We call the tuples $(s_j)=(s_1,\ldots,s_t)$ and $(\tau_j)=(\tau_1,\ldots,\tau_t)$ {\it words}.

Each $\ket{\phi_{k\ell}}$ is a normalized coherent sum over all pairs of
input words with $N_A=\ell$ on the left and $N_A=k$ on the right. The key
result is that \emph{this coherent sum has exactly the same Schmidt
spectrum as any single normalized term}. This reduction is proved in
Sec.~\ref{supp-subsec:word-reordering} of the
Supplemental Material (SM)~\cite{main-SupplementalMaterial}.

We now choose a representative that makes the imbalance in sector-label
counts between the two packets explicit. Let $m=|\ell-k|$
be the absolute difference between the numbers of $A$ labels in the two
words. For $\ell\ge k$, we choose $\ket{\Phi_t(u A^m|B^m u)}$, where
$u=A^kB^{t-\ell}$ is a common word of length $t-m$.

This common word can be deleted without changing the Schmidt spectrum:
the states $\ket{\Phi_t(u A^m|B^m u)}$ and $\ket{\Phi_m(A^m|B^m)}$
have the same Schmidt spectrum. This reduction is proved in
Sec.~\ref{supp-subsec:common-word-removal} of the SM.
We stress that when $u$ is nonempty,
the second state lives on a shorter
segment of $2m<2t$ sites and is generated by the smaller circuit $W_m$.
We call this reduced state the \emph{residual core}.

We now show that the residual cores
$\ket{\Phi_m(A^m|B^m)}$ and $\ket{\Phi_m(B^m|A^m)}$ are maximally
entangled in color space.
For $m=0$ the core is empty and the statement is immediate. For $m\ge1$,
consider first the vector
\begin{equation}
\ket{\Phi_m(A^m|B^m)}=W_m\left[  \ket{A,0}^{\otimes m}\otimes
   \ket{B,+}^{\otimes m}\right]  
\end{equation}
Now the sector labels are completely fixed: they are uniquely selected in the input state and they propagate
ballistically. The dynamics can thus be reduced to the evolution of the color part of the wave function. The left part
of the input state is a computational basis state in color space, but the right state involves linear combinations. We expand the
right input into $2^m$ basis states, and then evaluate the braiding steps for the individual pair of color words
$0^m$ and $b=(b_1,\ldots,b_m)$, with $b_j\in \Ftwo$.
For such inputs the computation is completely classical. The output of the braiding circuit $W_m$ can then be
described by two functions $\beta,\gamma: (\Ftwo)^m\to (\Ftwo)^m$, such that
\begin{equation}
  \begin{split}
  W_m\left[
\ket{0^m}_A\otimes \ket{b}_B
    \right]=
   \ket{\gamma(b)}_B\otimes \ket{\beta(b)}_A
  \end{split}
  \label{main-eq:state-color-maps}
 \end{equation}
Both $\beta$ and $\gamma$ are bijections, as proved in Sec.~\ref{supp-subsec:color-reconstruction} of the SM~\cite{main-SupplementalMaterial}. Expanding the right
input gives
\begin{equation}
  \label{main-eq:bell-core}
  \ket{\Phi_m(A^m|B^m)}=
\frac{1}{2^{m/2}}
\sum_{b_j\in \Ftwo}  \ket{\gamma(b)}_B\otimes \ket{\beta(b)}_A
\end{equation}
Bijectivity makes the output vectors orthonormal on each side, so the
sum over $b$ in Eq.~\eqref{main-eq:bell-core} is already a flat Schmidt
decomposition with $2^m$ equal coefficients. This proves maximal
entanglement. Exchanging the sector names gives the same conclusion
for the core $B^m|A^m$.

\head{Entanglement estimation}
We now control the coherent sum in Eq.~\eqref{main-eq:dw-branches}.
Distinct count values define orthogonal local sectors, but branches
sharing one count need not have orthogonal Schmidt vectors on that side.
Standard entropy bounds for convex combinations of density
matrices~\cite{main-NielsenChuang2010} give
\begin{equation}
 \begin{gathered}
 \overline S\le S_1^{\dw}(\Phi_t)\le\overline S+H,\\
 \overline S=\sum_{k,\ell}\pi^{\dw}_{k\ell}S_1(\phi_{k\ell}).
 \end{gathered}
 \label{main-eq:coherent-count-bound}
\end{equation}
Here $S_1^{\dw}$ is the bipartite von Neumann entropy and $\overline S$ is its
average over the normalized fixed-count branches, weighted by their
probabilities. The quantity $H$ is the Shannon entropy of the joint
distribution of the sector-label counts. The general statement and proof
are given in Sec.~\ref{supp-sec:coherent-count} of the SM~\cite{main-SupplementalMaterial}.

Each branch has entropy $S_1(\phi_{k\ell})=|k-\ell|\log2$.
Averaging $|k-\ell|$ with the joint weights $\pi^{\dw}_{k\ell}$ in
Eq.~\eqref{main-eq:dw-branches}, we obtain
\begin{equation}
 \overline S=\frac{\log2}{\sqrt\pi}\sqrt t+O(t^{-1/2}).
 \label{main-eq:state-mean}
\end{equation}
This binomial average is evaluated in Sec.~\ref{supp-subsec:binomial-vn} of the SM. Since $k,\ell\in\{0,\ldots,t\}$, their joint
distribution has at most $(t+1)^2$ outcomes. Its Shannon entropy therefore
satisfies $H\le\log (t+1)^2=2\log(t+1)$. Equation~\eqref{main-eq:coherent-count-bound}
therefore gives
\begin{equation}
 S_1^{\dw}(t)=\frac{\log2}{\sqrt\pi}\sqrt t+O(\log t).
 \label{main-eq:state-vn}
\end{equation}

\head{From states to operators} We now turn to operator entanglement.
The calculation closely parallels the domain-wall computation above. Only
a few modifications are required along the way, and we focus on these
differences.

We first use the folding trick to recast Heisenberg evolution as
Schr\"odinger evolution on a doubled Hilbert space~\cite{main-sajat-opent}.
We define vectorization by
$f(\ket{x}\bra{y})=\ket{x}_{\rm k}\ket{y}_{\rm b}$ and associate
the normalized vector $x_O=f(O)/\|O\|_{\HS}$ with any nonzero one-site
operator $O$.
Here $x,y$ label the full local basis, so the doubled local space has
dimension $16$. Vectorization is a local Hilbert--Schmidt isometry and
preserves the spatial bipartition and Schmidt coefficients. We need the
normalized vectorization of the identity:
\begin{equation}
 \iota=f(\id/2)=\frac12\sum_{s=A,B}\sum_{a\in\Ftwo}
 \ket{s,a}_{\rm k}\ket{s,a}_{\rm b}.
 \label{main-eq:folded-identities}
\end{equation}
The folded gate $\widehat R$ is defined by
$\widehat R f_2(X)=f_2(R^\dagger XR)$ with $f_2=f\otimes f$~\cite{main-sajat-opent}.
Since $R$ is real and Hermitian, it applies
the same $R$ independently to ket and bra. Unitarity gives the folded
counterpart of~\eqref{main-eq:bulk-fixed}:
\begin{equation}
 \widehat R(\iota\otimes\iota)=\iota\otimes\iota.
 \label{main-eq:folded-fixed}
\end{equation}
Let $\widehat W_t$ be obtained from $W_t$ by replacing each crossing $R$
with $\widehat R$.

For the local operator $O$ acting at site~$1$, the causal reduction
using~\eqref{main-eq:folded-fixed} leads to the auxiliary state~\cite{main-sajat-opent}
\begin{equation}
 \ket{\Omega_t}=\widehat W_t
 \left[\iota^{\otimes t}\otimes
 \left(x_O\otimes\iota^{\otimes(t-1)}\right)\right],
 \label{main-eq:op-rectangle}
\end{equation}
whose Schmidt spectrum across the outgoing packet cut equals the normalized
operator-Schmidt spectrum of $O(t)$ across $0|1$.
This has almost the same formal structure as the domain-wall problem in
Eq.~\eqref{main-eq:dw-rectangle}.

We now choose the sector-changing matrix unit
$O_\times=\ket{B,0}\bra{A,0}$.
The Hermitian, traceless source
$O_{\rm H}=(O_\times+O_\times^\dagger)/\sqrt2$ has exactly the same
normalized operator-Schmidt spectrum as $O_\times$ at every time, as proved
in Sec.~\ref{supp-subsec:hermitian-source} of the SM.
We therefore perform the calculation using $O_\times$.

In the folded circuit, each strand carries a ket--bra pair of sector
labels, which propagate ballistically. Identity inputs have equal sector
labels, $(A,A)$ or $(B,B)$, and include a sum over all colors, with the ket
and bra colors equal in each term. The source
has sectors $(B,A)$ and colors $(0,0)$. Thus the two layers start from
the same color data but evolve under different sector words.

For a fixed sector branch, let $k$ and $\ell$ be the numbers of $(A,A)$
pairs in the packet containing the source and in the other packet,
respectively, and put $d=\ell-k$.
In Sec.~\ref{supp-subsec:folded-reduction} of the SM
we show that the earlier reduction steps also apply to the folded circuit.
Removing the longest possible common word built from
the labels $(A,A)$ and $(B,B)$ leaves the residual cores
\begin{equation}
 \begin{cases}
 (A,A)^d\mid(B,A)(B,B)^{d-1},&d\ge1,\\
 (B,B)^{1-d}\mid(B,A)(A,A)^{-d},&d\le0.
 \end{cases}
 \label{main-eq:op-cores}
\end{equation}

For the $d\ge1$ core, let
$\iota_s=2^{-1/2}\sum_{a\in\Ftwo}\ket{s,a}_{\rm k}\ket{s,a}_{\rm b}$
be the normalized identity input in sector $s=A,B$.
Having identified the residual core, we need to evaluate
\begin{equation}
 \widehat W_d
 \left[\iota_A^{\otimes d}\otimes
 \left(x_{O_\times}\otimes\iota_B^{\otimes(d-1)}\right)\right].
 \label{main-eq:op-core-state}
\end{equation}
Expanding the identity inputs gives a sum over $d$ input colors on the
left and $d-1$ on the right, with both source colors fixed to zero.
We show in Sec.~\ref{supp-subsec:folded-reconstruction} of the SM~\cite{main-SupplementalMaterial}
that this state has a flat Schmidt spectrum with Schmidt rank $2^{d-1}$.
For $d\le0$, exchanging ket and bra gives a flat spectrum with rank
$2^{-d}$. Thus the entropy of each normalized branch is
$|\ell-k|\log2$ for $\ell\le k$ and $(|\ell-k|-1)\log2$ for $\ell>k$.
We absorb this bounded correction into the remainder below, since it
does not alter the asymptotic growth laws.

We now restore the coherent sum. As in the state computation, deterministic
transport places the $(k,\ell)$ branches in orthogonal local output count
sectors, so the entropy bound in Eq.~\eqref{main-eq:coherent-count-bound}
applies to $\ket{\Omega_t}$.
Equation~\eqref{main-eq:folded-identities} gives independent unbiased binomial
counts $K$ and $L$, with $t-1$ and $t$ trials, respectively.
Averaging the branch entropies gives
\begin{equation}
 \overline S=(\log2)\E|L-K|+O(1)
 =\frac{\log2}{\sqrt\pi}\sqrt t+O(1).
 \label{main-eq:op-mean}
\end{equation}
The binomial average is evaluated in Sec.~\ref{supp-subsec:binomial-vn} of the SM~\cite{main-SupplementalMaterial}.
The Shannon entropy $H$ of the joint count distribution obeys
$H\le\log[t(t+1)]$, since there are $t(t+1)$ possible count pairs.
Combining these estimates with Eq.~\eqref{main-eq:coherent-count-bound} yields
\begin{equation}
 S_1^\op(t)=\frac{\log2}{\sqrt\pi}\sqrt t+O(\log t).
 \label{main-eq:op-vn}
\end{equation}
These bounds establish superlogarithmic growth of the operator entanglement.

\head{R\'enyi entropies}
We also derive bounds on the R\'enyi entropies and find markedly different behavior
for fixed $\alpha>0$, $\alpha\ne1$:
\begin{equation}
 S_\alpha^\op(t)=
 \begin{cases}
 v_\alpha t+O(\log t),&0<\alpha<1,\\
 \Theta(\log t),&1<\alpha,
 \end{cases}
 \label{main-eq:hierarchy}
\end{equation}
where $\Theta(\log t)$ means that the entropy is bounded above and below
by positive multiples of $\log t$ for sufficiently large $t$.
The growth rate $v_\alpha$ is derived
in Sec.~\ref{supp-sec:renyi} of the SM~\cite{main-SupplementalMaterial}.
The earlier domain-wall problem obeys the same growth laws.

The linear growth for $0<\alpha<1$ reflects the enhanced contribution of
exponentially many small Schmidt probabilities, associated with rare
imbalances of order $t$. The von Neumann entropy probes the entanglement
generated by typical imbalances of order $\sqrt t$. For $\alpha>1$, the
larger Schmidt probabilities are weighted more strongly, and their scale
is set by rare, nearly balanced branches with $|d|=O(1)$.

\head{MPO approximation}
The entropy scaling laws above do not by themselves determine the bond
dimension needed for tensor network simulations at fixed approximation
accuracy~\cite{main-SchuchWolfVerstraeteCirac2008}.
This depends on the sum of the largest Schmidt probabilities.
We ask how the bond dimension required to approximate $O(t)$ by a matrix
product operator (MPO) grows at fixed accuracy. Let $\chi_\varepsilon(t)$
be the least bond dimension across the cut $0|1$ needed for a relative
squared Hilbert--Schmidt error at most $\varepsilon$, with
$0<\varepsilon<1$ fixed. There exists $c_\varepsilon>0$ such that
$\chi_\varepsilon(t)\ge\exp(c_\varepsilon\sqrt t)$ for all sufficiently
large $t$.
This bound obstructs efficient MPO simulation: MPOs with polynomially growing
bond dimension capture an asymptotically vanishing fraction of the operator's
\mbox{Hilbert--Schmidt} weight.
The proof is given in Sec.~\ref{supp-sec:mpo} of the SM~\cite{main-SupplementalMaterial}.

\head{Discussion}
Our main result is that
Yang--Baxter integrability does not enforce logarithmic growth of the
operator entanglement, or polynomial bond dimension at fixed 
Hilbert--Schmidt accuracy.
We presented one counterexample with a special mechanism for entanglement production. The
circuit fits into the standard Yang--Baxter framework.

The distinguished properties were made possible by choosing four-dimensional local spaces
and a special local permutation action on the computational basis states.
The gate propagates the sector labels
ballistically, and it can generate maximal entanglement in the color variables
if the incoming sector labels differ. At the same time, the gate acts
identically on the color variables if the incoming sector labels are identical.
This also means that the gate has invariant subspaces, and this is a special
type of decomposability \cite{main-MajidMarkl1996,main-GatevaIvanovaMajid2008,main-Nichita2013,main-GomborPozsgay2022}.

The gate is not invariant under spatial reflection. However, we conjecture that this is not a limitation. In
Eq.~\eqref{supp-eq:reflection-gate} of the SM we introduce a
reflection-invariant $D=8$ circuit. We conjecture that it has the same
entanglement-growth scaling forms.

The results leave us with the following open questions: What is the minimal local dimension such that the square-root law can be
realized? Is the square-root law the largest possible von Neumann entropy production for an integrable model? How
often does the square-root law appear among integrable models? Can we exclude the square-root law by a proper
indecomposability condition? We hope to return to these questions.


\begin{acknowledgments}
\head{Acknowledgments}
We are grateful to Istv\'an Vona for useful discussions and collaboration
during the early stages of this project.
\end{acknowledgments}

\paragraph*{Data availability.} No new numerical data are reported in this work.

%

\makeatletter
\close@column
\clearpage
\onecolumngrid
\setcounter{page}{1}
\@booleanfalse\twocolumn@sw
\def\title@column#1{\minipagefootnote@init#1\minipagefootnote@foot}
\def\close@column{\newpage}
\def\arxiv@part{supp}
\frontmatter@init
\let\author\frontmatter@author
\let\affiliation\arxiv@affiliation
\let\and\frontmatter@and
\let\maketitle\frontmatter@maketitle
\let\arxiv@label\label
\def\label#1{\def\arxiv@key{#1}\def\arxiv@first{FirstPage}%
  \ifx\arxiv@key\arxiv@first\arxiv@label{supp-FirstPage}%
  \else\arxiv@label{#1}\fi}
\makeatother
\setcounter{section}{0}
\setcounter{subsection}{0}
\setcounter{subsubsection}{0}
\setcounter{equation}{0}
\setcounter{figure}{0}
\setcounter{table}{0}
\setcounter{footnote}{0}
\setcounter{secnumdepth}{3}
\renewcommand{\theequation}{S\arabic{equation}}
\renewcommand{\thefigure}{S\arabic{figure}}
\renewcommand{\thetable}{S\arabic{table}}
\renewcommand{\theHequation}{supp.\arabic{equation}}
\renewcommand{\theHfigure}{supp.\arabic{figure}}
\renewcommand{\theHtable}{supp.\arabic{table}}
\renewcommand{\theHsection}{supp.\arabic{section}}
\renewcommand{\theHsubsection}{supp.\arabic{section}.\arabic{subsection}}

\title{Supplemental Material:\texorpdfstring{\\}{ }
Square-root growth of operator entanglement in an integrable brickwork circuit}
\author{Bal\'azs Pozsgay}
\affiliation{MTA-ELTE ``Momentum'' Integrable Quantum Dynamics Research Group,\protect\\
ELTE E\"otv\"os Lor\'and University, Budapest, Hungary}
\date{\today}
\maketitle

\section{Schmidt spectra in the domain-wall problem}
\label{supp-sec:state-spectrum}

In this section we provide details supporting the claim that the
normalized fixed-count states $\ket{\phi_{k\ell}}$ introduced in
Eq.~\eqref{main-eq:dw-branches} of the main text have bipartite
entanglement $S_1(\phi_{k\ell})=|k-\ell|\log2$.

\subsection{Word reordering in fixed-count branches}
\label{supp-subsec:word-reordering}

We start by expanding $\ket{\phi_{k\ell}}$ in states with fixed incoming
sector labels.
Recall that the sector labels propagate ballistically: the entangling
rectangle $W_t$ exchanges the two incoming packets while preserving the
order of the labels within each packet. As in the main text, we call an
ordered sequence of sector labels, such as $(s_j)=(s_1,\ldots,s_t)\in\{A,B\}^t$, a
\emph{word}, and write $N_A((s_j))$ for its number of $A$ letters.
Expanding Eq.~\eqref{main-eq:count-transport} of the main text gives
\begin{equation}
  \ket{\phi_{k\ell}}
  =\frac{1}{\sqrt{\binom{t}{\ell}\binom{t}{k}}}
  \sum_{\substack{(s_j):\,N_A((s_j))=\ell\\
                  (\tau_j):\,N_A((\tau_j))=k}}
  \ket{\Phi_t((s_j)|(\tau_j))}.
  \label{supp-eq:fixed-count-word-sum}
\end{equation}
Here $\ket{\Phi_t((s_j)|(\tau_j))}$ is the output of $W_t$ for
the prescribed sector words $(s_j)$ and $(\tau_j)$ on the left and
right, respectively. The input color state is $\ket{0}$ at every left
site and $\ket{+}$ at every right site:
\begin{equation}
  \ket{\Phi_t((s_j)|(\tau_j))}
  =W_t\left[
  \left(\bigotimes_{j=1}^t\ket{s_j,0}\right)\otimes
  \left(\bigotimes_{j=1}^t\ket{\tau_j,+}\right)
  \right].
  \label{supp-eq:fixed-word-state}
\end{equation}
The input is not a computational basis state, since each $\ket{+}$ is a
superposition in color space.

The sum in Eq.~\eqref{supp-eq:fixed-count-word-sum} contains
$\binom{t}{\ell}\binom{t}{k}$ mutually orthogonal
terms, which explains the normalization. We now show that the coherent
sum in Eq.~\eqref{supp-eq:fixed-count-word-sum} has the same nonzero Schmidt
spectrum as any one of its fixed-word terms.

\begin{unnumberedlemma}[Fixed-word reordering]
Let $(s_j),(\tau_j)$ and
$(\tilde s_j),(\tilde\tau_j)$ be two pairs of words with the same
counts $(\ell,k)$. Then
$\ket{\Phi_t((s_j)|(\tau_j))}$ and
$\ket{\Phi_t((\tilde s_j)|(\tilde\tau_j))}$ have the same Schmidt
spectrum.
\end{unnumberedlemma}

\begin{proof}
There are permutations $p,q\in S_t$ such that
$(\tilde s_j)=p(s_j)$ and
$(\tilde\tau_j)=q(\tau_j)$. The permutations need not be unique.
The SWAP action established in the main text implies that the action of the unitary operators $U_p$ and $U_q$ on the
input vectors realizes
the permutations $p$ and $q$ on the words, without changing the color part of the wave function.
Repeated use of the braid relation moves the two input operations
through the rectangle and exchanges their packet positions:
\begin{equation}
  W_t (U_p\otimes U_q)=(U_q\otimes U_p) W_t
\end{equation}
This relation shows that permuting the input words is equivalent to
separate unitaries on the left and right outputs. Since these unitaries
act locally with respect to the bipartition, they preserve the Schmidt
spectrum, which proves the claim.
\end{proof}

\begin{unnumberedlemma}[Fixed-count branch reduction]
For fixed $k$ and $\ell$, the state $\ket{\phi_{k\ell}}$ has the same
nonzero Schmidt spectrum as every fixed-word state appearing in
Eq.~\eqref{supp-eq:fixed-count-word-sum}.
\end{unnumberedlemma}

\begin{proof}
For each
distinct word with the prescribed count, choose one permutation taking
the reference word to that word. We sum once over each distinct word,
not over all permutations realizing it. On each output packet, restrict
the corresponding unitaries to the subspace with the reference sector
word and arbitrary colors. Different terms have orthogonal sector-word
images, and each term preserves norms. Dividing each sum by the square
root of its number of terms therefore gives an isometry on that
fixed-word subspace. The two sums factorize across the bipartition and
map the reference state to the normalized fixed-count branch. This branch
thus has the same nonzero Schmidt spectrum as the chosen fixed-word
state.
\end{proof}

\subsection{Permutations fixing a sector word}
\label{supp-subsec:word-fixing}

We introduce an auxiliary lemma, which will be used in the proof of the
common-word-removal lemma in Sec.~\ref{supp-subsec:common-word-removal} below.

\begin{unnumberedlemma}[Word fixing]
Let $\ket{\Psi}$ be a vector on $2t$ sites whose sector labels form a
fixed word $v\in\{A,B\}^{2t}$; its color state may be arbitrary. If
$p\in S_{2t}$ fixes $v$, then $U_p\ket{\Psi}=\ket{\Psi}$.
\end{unnumberedlemma}

\begin{proof}
Choose a permutation $q$ such that
$qv=A^rB^{2t-r}$, where $r=N_A(v)$. Then $p'=qpq^{-1}$ fixes
$A^rB^{2t-r}$ and can be written as $p'=p'_1p'_2$, where $p'_1$ and
$p'_2$ act on the first $r$ and last $2t-r$ elements, respectively.
Using the symmetric-group representation established in the main text,
we obtain $U_p=U_q^\dagger U_{p'}U_q$. The vector $U_q\ket{\Psi}$ has
fixed sector word $A^rB^{2t-r}$. Within each equal-sector block, every
adjacent generator acts as the identity by
Eq.~\eqref{main-eq:gate} of the main text. Consequently
$U_{p'}=U_{p'_1}U_{p'_2}$ fixes $U_q\ket{\Psi}$, and conjugating back
gives $U_p\ket{\Psi}=\ket{\Psi}$.
\end{proof}

\subsection{Common-word removal}
\label{supp-subsec:common-word-removal}

We now use the reordering freedom of Sec.~\ref{supp-subsec:word-reordering}
to isolate the common sector content of the two input words. Recall
that the left and right words contain $\ell$ and $k$ letters $A$,
respectively, and set $m=|k-\ell|$. We choose the common word
$u=A^{\min(k,\ell)}B^{t-\max(k,\ell)}$ by taking the smaller of the two
input counts for each sector; its length is $t-m$. We write
$\overset{\mathrm{Sch}}{\sim}$ when two states have the same nonzero
Schmidt spectrum across the packet cut.

\begin{unnumberedlemma}[Common-word removal]
The common word $u$ can be removed without changing this spectrum:
\begin{equation}
  \begin{aligned}
  \ket{\Phi_t(uA^m|B^mu)}
  &\overset{\mathrm{Sch}}{\sim}
  \ket{\Phi_m(A^m|B^m)}, && \ell\ge k,\\
  \ket{\Phi_t(uB^m|A^mu)}
  &\overset{\mathrm{Sch}}{\sim}
  \ket{\Phi_m(B^m|A^m)}, && k>\ell.
  \end{aligned}
  \label{supp-eq:common-word-removal}
\end{equation}
\end{unnumberedlemma}

\begin{proof}
We now apply the word-fixing lemma from
Sec.~\ref{supp-subsec:word-fixing}.
Take the initial words $u A^m|B^m u$ and first apply the inner circuit
$W_m$, so that the combined word becomes $u B^m|A^m u$. This circuit
acts only on the unequal central words. The two copies of $u$ remain
untouched product factors on their respective sides of the cut, so the
Schmidt spectrum is that of $\ket{\Phi_m(A^m|B^m)}$.
The full circuit $W_t$ contains further braid operations. Those exchanging
$uB^m\to B^mu$ and $A^mu\to uA^m$ are local with respect to the
bipartition. The remaining circuit $W_{t-m}$ exchanges the two copies
of $u$ in the center. Its permutation $w_{t-m}$ fixes the sector word
$u|u$. By the word-fixing lemma, this remaining circuit acts trivially
on the state. This proves the statement for $\ell\ge k$.
For $\ell<k$, exchanging the sector names leaves the core
$B^m|A^m$ with $m=k-\ell$.
\end{proof}

\subsection{Reconstruction of colors}
\label{supp-subsec:color-reconstruction}

Our goal is to prove that the residual core
$\ket{\Phi_m(A^m|B^m)}$ is maximally entangled in color space across the
two packets of length $m$. It will then have $2^m$ equal Schmidt
coefficients and entanglement entropy
$\log(2^m)=m\log2$.

For $b=(b_1,\ldots,b_m)\in(\Ftwo)^m$, the main text defined the maps
$\beta,\gamma:(\Ftwo)^m\to(\Ftwo)^m$ via
\begin{equation}
  \begin{split}
  W_m\left[
  \ket{0^m}_A\otimes \ket{b}_B
  \right]
  =\ket{\gamma(b)}_B\otimes \ket{\beta(b)}_A.
  \end{split}
  \label{supp-eq:state-color-maps-supp}
\end{equation}
We write
$\gamma(b)=(\gamma_1,\ldots,\gamma_m)$ and
$\beta(b)=(\beta_1,\ldots,\beta_m)$. Expanding the right input over all
color words gives
\begin{equation}
  \ket{\Phi_m(A^m|B^m)}=
  \frac{1}{2^{m/2}}
  \sum_{b_j\in\Ftwo}
  \ket{\gamma(b)}_B\otimes\ket{\beta(b)}_A.
  \label{supp-eq:bell-core-supp}
\end{equation}
If $\beta$ and $\gamma$ are bijections, the two families of output
vectors in Eq.~\eqref{supp-eq:bell-core-supp} are orthonormal. The displayed
sum is then a flat Schmidt decomposition, so proving bijectivity is
enough to establish maximal entanglement. For a fixed basis word $b$,
$W_m$ acts as a permutation on fixed sector and color labels. The color
calculation used to prove this bijectivity is therefore completely
classical.

Figure~\ref{supp-fig:state-color-rectangle}(a) shows the color calculation for
$m=3$ with fixed sector labels. We label the incoming $A$ strands on the
left by $A_1,\ldots,A_m$ and the incoming $B$ strands on the right by
$B_1,\ldots,B_m$, in left-to-right order within each input packet. These
strand labels are written below the input colors; they follow the
strands through the braid, while the colors change at crossings. Each
pair $A_i,B_j$ crosses once, so the pair $(i,j)$ uniquely identifies a
crossing. We evaluate the circuit by following $B_1,\ldots,B_m$ in
order; each $B_j$ crosses $A_m,\ldots,A_1$ in that order and exits on
the left.

Figure~\ref{supp-fig:state-color-rectangle}(b) makes the local rule explicit:
inputs $(A,a;B,c)$ give outputs $(B,h(a,c);A,c)$, with $h(a,c)=a+c$
by Eq.~\eqref{main-eq:gate} of the main text.

The crucial property of the function $h$ is that it is invertible in either argument when the other one
is fixed.
This means that knowing $h(a,c)$ and either
input color determines the other input color uniquely. In
Fig.~\ref{supp-fig:state-color-rectangle}(b), the two left-hand colors $a$
and $h(a,c)$ therefore determine $c$. This local reconstruction resembles
that permitted by dual unitarity, and it is used in the proof below. However, the gate is not dual-unitary, because the 
opposite diagonal reconstruction
fails. The right-hand colors in Fig.~\ref{supp-fig:state-color-rectangle}(b) are both $c$;
therefore, these right-diagonal data do not determine all color variables.

\definecolor{reconfirst}{HTML}{0072B2}
\definecolor{reconsecond}{HTML}{D55E00}
\definecolor{reconthird}{HTML}{009E73}
\newcommand{\colorrectangle}[1]{%
\begin{tikzpicture}[x=1.30cm,y=0.72cm,line cap=round,line join=round,
 font=\small,
 colorwire/.style={draw=black!78,line width=.75pt},
 colorvertex/.style={circle,draw=blue!65!black,fill=white,
   minimum size=2.1mm,inner sep=0pt,line width=.7pt},
 colorlabel/.style={font=\small,inner sep=1pt},
 strandlabel/.style={font=\footnotesize,text=black!70,inner sep=1pt},
 reconstruction/.style={line width=1.75pt,postaction={decorate},
   decoration={markings,mark=at position .54 with
     {\arrow{Stealth[length=2.8mm,width=2.1mm]}}}}]
 \foreach \x in {0,...,5}\draw[colorwire](\x,0)--(\x,.50);
 \foreach \x in {0,1,4,5}\draw[colorwire](\x,.50)--(\x,1.40);
 \foreach \x/\xp in {2/3,3/2}\draw[colorwire](\x,.50)--(\xp,1.40);
 \foreach \x in {0,5}\draw[colorwire](\x,1.40)--(\x,2.30);
 \foreach \x/\xp in {1/2,2/1,3/4,4/3}
   \draw[colorwire](\x,1.40)--(\xp,2.30);
 \foreach \x/\xp in {0/1,1/0,2/3,3/2,4/5,5/4}
   \draw[colorwire](\x,2.30)--(\xp,3.20);
 \foreach \x in {0,5}\draw[colorwire](\x,3.20)--(\x,4.10);
 \foreach \x/\xp in {1/2,2/1,3/4,4/3}
   \draw[colorwire](\x,3.20)--(\xp,4.10);
 \foreach \x in {0,1,4,5}\draw[colorwire](\x,4.10)--(\x,5.00);
 \foreach \x/\xp in {2/3,3/2}\draw[colorwire](\x,4.10)--(\xp,5.00);
 \foreach \x in {0,...,5}\draw[colorwire](\x,5.00)--(\x,5.55);

 #1
 \foreach \x/\y in {2.5/.95,1.5/1.85,3.5/1.85,
   .5/2.75,2.5/2.75,4.5/2.75,1.5/3.65,3.5/3.65,2.5/4.55}
   \node[colorvertex] at (\x,\y){};

 \foreach \x/\i in {0/1,1/2,2/3}{
   \node[colorlabel,anchor=north] at (\x,-.08) {$0$};
   \node[strandlabel,anchor=north] at (\x,-.58) {$A_{\i}$};
 }
 \foreach \x/\i in {3/1,4/2,5/3}{
   \node[colorlabel,anchor=north] at (\x,-.08) {$b_{\i}$};
   \node[strandlabel,anchor=north] at (\x,-.58) {$B_{\i}$};
 }
 \foreach \x/\i in {0/1,1/2,2/3}
   \node[colorlabel,anchor=south] at (\x,5.65) {$\gamma_{\i}$};
 \foreach \x/\i in {3/1,4/2,5/3}
   \node[colorlabel,anchor=south] at (\x,5.65) {$\beta_{\i}$};
 \draw[dashed,black!45](2.5,5.15)--(2.5,6.15);
\end{tikzpicture}
}

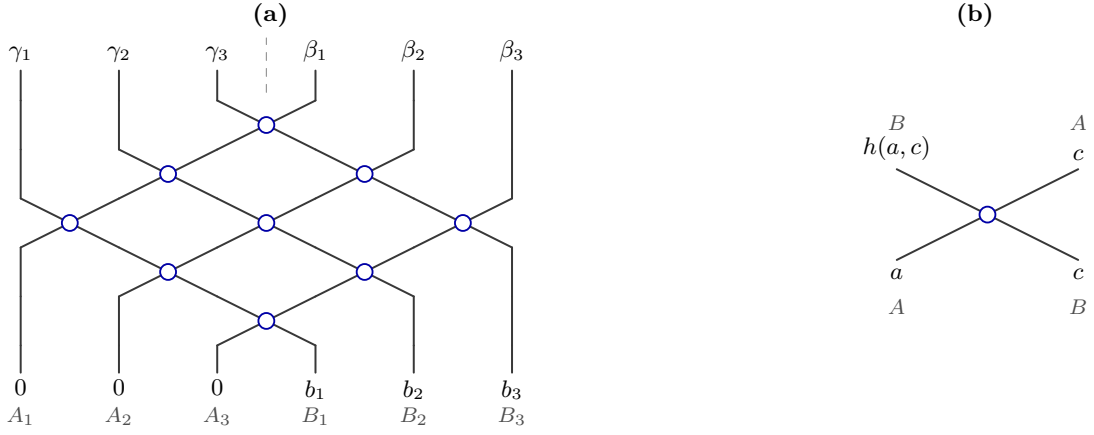
\begin{figure}[!tb]
\centering
\begin{minipage}[t]{.60\linewidth}
\centering
\textbf{(a)}\par\smallskip
\colorrectangle{}
\end{minipage}\hfill
\begin{minipage}[t]{.36\linewidth}
\centering
\textbf{(b)}\par\vspace{1.1cm}
\begin{tikzpicture}[x=1cm,y=1cm,line cap=round,line join=round,
 font=\small,
 colorwire/.style={draw=black!78,line width=.75pt},
 colorvertex/.style={circle,draw=blue!65!black,fill=white,
   minimum size=2.1mm,inner sep=0pt,line width=.7pt},
 colorlabel/.style={font=\small,inner sep=1pt},
 sectorlabel/.style={font=\footnotesize,text=black!70,inner sep=1pt}]
 \draw[colorwire] (-1.2,0)--(1.2,1.2);
 \draw[colorwire] (1.2,0)--(-1.2,1.2);
 \node[colorvertex] at (0,.6){};
 \node[colorlabel,anchor=north] at (-1.2,-.08) {$a$};
 \node[colorlabel,anchor=north] at (1.2,-.08) {$c$};
 \node[sectorlabel,anchor=north] at (-1.2,-.48) {$A$};
 \node[sectorlabel,anchor=north] at (1.2,-.48) {$B$};
 \node[colorlabel,anchor=south] at (-1.2,1.28) {$h(a,c)$};
 \node[colorlabel,anchor=south] at (1.2,1.28) {$c$};
 \node[sectorlabel,anchor=south] at (-1.2,1.68) {$B$};
 \node[sectorlabel,anchor=south] at (1.2,1.68) {$A$};
\end{tikzpicture}
\end{minipage}
\caption{(a) Color dynamics of the residual rectangle for $m=3$. The
sector label carried by each strand is fixed to $A$ or $B$ and is
denoted in the diagram. The variables $b_i$, $\gamma_i$, and $\beta_i$
denote colors, and all left input colors are fixed to zero. At each lower
endpoint the strand label is written below the input color. The dashed
line marks the cut defining the bipartition used to compute the
entanglement. (b) A single crossing with fixed
sector labels: input colors $(a,c)$ give output colors $(h(a,c),c)$,
where $h(a,c)=a+c$ in $\Ftwo$. Time runs upwards in both panels.}
\label{supp-fig:state-color-rectangle}
\end{figure}

\begin{unnumberedlemma}[Bijection of the residual color maps]
For every $m\ge1$, the maps $\beta$ and $\gamma$ defined in
Eq.~\eqref{supp-eq:state-color-maps-supp} are bijections. Consequently,
Eq.~\eqref{supp-eq:bell-core-supp} is a flat Schmidt decomposition and
$\ket{\Phi_m(A^m|B^m)}$ has entanglement entropy $m\log2$. The same
conclusion holds for $\ket{\Phi_m(B^m|A^m)}$.
\end{unnumberedlemma}

\begin{proof}
We prove the two bijections by reconstructing the input color word $b$
from each output packet separately. Figure~\ref{supp-fig:state-color-reconstruction}
shows these two independent arguments for $m=3$: panel (a) uses the right
output $\beta(b)$, while panel (b) uses the left output $\gamma(b)$
together with the fixed left input colors $0^m$. The thick colored paths
show the order in which intermediate colors in the circuit are
reconstructed until all entries of $b$ are known. Arrows indicate the
direction of reconstruction. The proof consists of showing that every
local reconstruction along these paths can be performed using colors
already known at that step.

We now explain both reconstructions in words, following the order shown
in the figure. Let $x_i^{(j)}$ be the color
on $A_i$ after its crossing with $B_j$, with $x_i^{(0)}=0$ the input
color. Thus $x^{(j)}=(x_1^{(j)},\ldots,x_m^{(j)})$ is the working
$A$-packet after the first $j$ incoming $B$ strands have passed. The strand
$B_j$ enters with color $b_j$ and exits with color $\gamma_j$, while
$\beta_i=x_i^{(m)}$.
Following the crossings along $B_j$ gives
\begin{align}
 \gamma_j&=h\bigl(x_1^{(j-1)},x_1^{(j)}\bigr),\nonumber\\
 x_{i-1}^{(j)}&=h\bigl(x_i^{(j-1)},x_i^{(j)}\bigr)
       \quad(2\le i\le m),\label{supp-eq:local-peeling}\\
 x_m^{(j)}&=b_j.\nonumber
\end{align}
We first reconstruct $b$ from the left output, as in
Fig.~\ref{supp-fig:state-color-reconstruction}(b). Starting from the known
packet $x^{(0)}=0^m$, the outgoing color $\gamma_1$ determines
$x_1^{(1)}$. Reading the remaining crossings
backwards along $B_1$, from $A_2$ to $A_m$, determines the other entries
of $x^{(1)}$ and finally gives $b_1$. We then repeat this reconstruction
along $B_2,\ldots,B_m$, using the packet obtained at the preceding step.
This recovers $b$ from the left output $\gamma(b)$.

We now turn to the independent reconstruction from the right output,
shown in Fig.~\ref{supp-fig:state-color-reconstruction}(a). For $m=2$, the
reconstruction order can be checked directly.
The final entry gives $x_2^{(2)}=b_2$. The middle relation in
Eq.~\eqref{supp-eq:local-peeling} then reads
$x_1^{(2)}=h(x_2^{(1)},b_2)=h(b_1,b_2)$, which determines $b_1$
from the other final entry. Thus we recover $b_2$ before $b_1$.

For general $m$, the final packet $\beta(b)=x^{(m)}$ supports the same
backward reconstruction. Its last entry is $b_m$. For $m\ge2$, the middle relation in
Eq.~\eqref{supp-eq:local-peeling} reconstructs the colors on $A_2,\ldots,A_m$
before they crossed $B_m$, namely $(x_2^{(m-1)},\ldots,x_m^{(m-1)})$.
The last of these colors is $b_{m-1}$. Continuing backwards through the
successive $B$ strands reconstructs successively shorter suffixes of the
$A$-packet and recovers $b_m,b_{m-1},\ldots,b_1$. Each output therefore
retains the complete input color word. Thus $\beta$ and $\gamma$ are
injective, and hence bijective because their domain and codomain have the
same finite cardinality.

For the core $B^m|A^m$, exchanging the sector names leaves the local
relation $h(a,c)=a+c$ unchanged, so the same reconstruction proves
bijectivity.

Substituting the two bijections into Eq.~\eqref{supp-eq:bell-core-supp}
makes both sets of Schmidt vectors orthonormal. Its $2^m$ coefficients
are all $2^{-m/2}$, which proves the stated entropy and completes the
proof.
\end{proof}

\begin{figure}[!tb]
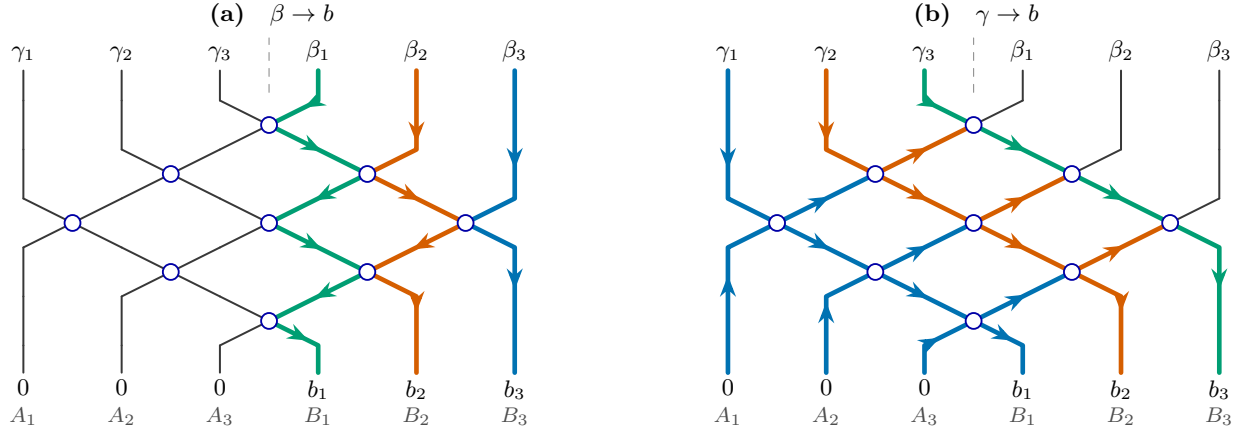

\centering
\begin{minipage}[t]{.48\linewidth}
\centering
\textbf{(a)}\quad $\beta\to b$\par\smallskip
\colorrectangle{%
 \draw[reconstruction,reconfirst] (5,5.55)--(5,3.20)--(4.5,2.75);
 \draw[reconstruction,reconfirst] (4.5,2.75)--(5,2.30)--(5,0);
 \draw[reconstruction,reconsecond] (4,5.55)--(4,4.10)--(3.5,3.65);
 \draw[reconstruction,reconsecond] (3.5,3.65)--(4.5,2.75);
 \draw[reconstruction,reconsecond] (4.5,2.75)--(3.5,1.85);
 \draw[reconstruction,reconsecond] (3.5,1.85)--(4,1.40)--(4,0);
 \draw[reconstruction,reconthird] (3,5.55)--(3,5.00)--(2.5,4.55);
 \draw[reconstruction,reconthird] (2.5,4.55)--(3.5,3.65);
 \draw[reconstruction,reconthird] (3.5,3.65)--(2.5,2.75);
 \draw[reconstruction,reconthird] (2.5,2.75)--(3.5,1.85);
 \draw[reconstruction,reconthird] (3.5,1.85)--(2.5,.95);
 \draw[reconstruction,reconthird] (2.5,.95)--(3,.50)--(3,0);
}
\end{minipage}\hfill
\begin{minipage}[t]{.48\linewidth}
\centering
\textbf{(b)}\quad $\gamma\to b$\par\smallskip
\colorrectangle{%
 \draw[reconstruction,reconfirst] (0,0)--(0,2.30)--(.5,2.75);
 \draw[reconstruction,reconfirst] (1,0)--(1,1.40)--(1.5,1.85);
 \draw[reconstruction,reconfirst] (2,0)--(2,.50)--(2.5,.95);
 \draw[reconstruction,reconfirst] (.5,2.75)--(1.5,3.65);
 \draw[reconstruction,reconfirst] (1.5,1.85)--(2.5,2.75);
 \draw[reconstruction,reconfirst] (2.5,.95)--(3.5,1.85);
 \draw[reconstruction,reconsecond] (1.5,3.65)--(2.5,4.55);
 \draw[reconstruction,reconsecond] (2.5,2.75)--(3.5,3.65);
 \draw[reconstruction,reconsecond] (3.5,1.85)--(4.5,2.75);
 \draw[reconstruction,reconfirst] (0,5.55)--(0,3.20)--(.5,2.75);
 \draw[reconstruction,reconfirst] (.5,2.75)--(1.5,1.85);
 \draw[reconstruction,reconfirst] (1.5,1.85)--(2.5,.95);
 \draw[reconstruction,reconfirst] (2.5,.95)--(3,.50)--(3,0);
 \draw[reconstruction,reconsecond] (1,5.55)--(1,4.10)--(1.5,3.65);
 \draw[reconstruction,reconsecond] (1.5,3.65)--(2.5,2.75);
 \draw[reconstruction,reconsecond] (2.5,2.75)--(3.5,1.85);
 \draw[reconstruction,reconsecond] (3.5,1.85)--(4,1.40)--(4,0);
 \draw[reconstruction,reconthird] (2,5.55)--(2,5.00)--(2.5,4.55);
 \draw[reconstruction,reconthird] (2.5,4.55)--(3.5,3.65);
 \draw[reconstruction,reconthird] (3.5,3.65)--(4.5,2.75);
 \draw[reconstruction,reconthird] (4.5,2.75)--(5,2.30)--(5,0);
}
\end{minipage}
\caption{Two independent reconstructions of the input color sequence
$(b_1,b_2,b_3)$ from (a) the right output $\beta$ and (b) the left output
$\gamma$ together with the fixed left input colors $0^3$. Arrows show
the reconstruction direction.}
\label{supp-fig:state-color-reconstruction}
\end{figure}

\section{Schmidt spectra in the operator time-evolution problem}
\label{supp-sec:operator-spectrum}

We now turn to operator entanglement and introduce notation that directly
generalizes the fixed-word states and count branches used in the
domain-wall problem. Our starting point is Eq.~\eqref{main-eq:op-rectangle}
of the main text:
\begin{equation}
 \ket{\Omega_t}=\widehat W_t
 \left[\iota^{\otimes t}\otimes
 \left(x_O\otimes\iota^{\otimes(t-1)}\right)\right].
 \label{supp-eq:folded-rectangle-supp}
\end{equation}
Here $\widehat W_t$ is the folded rectangle, $\iota=f(\id/2)$ is the
normalized vectorized identity, and $x_O=f(O)/\|O\|_{\HS}$ is the
normalized source. The Schmidt spectrum across the outgoing packet cut
is the normalized operator-Schmidt spectrum of $O(t)$. We perform the
calculation for $O_\times=\ket{B,0}\bra{A,0}$, whose vectorization is
$x_{O_\times}=\ket{B,0}_{\rm k}\ket{A,0}_{\rm b}$.

Recall that each site in the folded circuit carries a pair of sector
labels, one for the ket and one for the bra, which propagate together
ballistically along a strand. In principle, there are four possible
pairs: $(A,A)$, $(A,B)$, $(B,A)$, and $(B,B)$. In the computation for
$O_\times$, only three sector pairs occur: $(A,A)$ and $(B,B)$ at the
identity inputs, and the single source pair $(B,A)$.
A \emph{word} in the folded problem is an ordered sequence of
these label pairs. For the free labels, we use the shorter notation
$A$ and $B$ for $(A,A)$ and $(B,B)$, respectively; the source pair is
fixed and is not included among the free labels.

To expand Eq.~\eqref{supp-eq:folded-rectangle-supp} over sectors, we resolve
each identity input into its two normalized components:
\begin{equation}
 \iota=\frac{\iota_A+\iota_B}{\sqrt2},\qquad
 \iota_s=\frac{1}{\sqrt2}\sum_{c\in\Ftwo}
 \ket{s,c}_{\rm k}\ket{s,c}_{\rm b},\qquad s=A,B.
 \label{supp-eq:folded-sector-inputs}
\end{equation}
Thus the free sector labels are summed independently in the two input
packets, giving
\begin{equation}
 \ket{\Omega_t}
 =\frac{1}{2^{t-1/2}}
 \sum_{(s_j)\in\{A,B\}^t}
 \sum_{(\tau_j)\in\{A,B\}^{t-1}}
 \widehat W_t\left[
 \left(\bigotimes_{j=1}^t\iota_{s_j}\right)\otimes
 \left(x_{O_\times}\otimes
       \bigotimes_{j=1}^{t-1}\iota_{\tau_j}\right)
 \right].
 \label{supp-eq:folded-sector-expansion}
\end{equation}
The left packet has $t$ free labels, while the right has $t-1$ because
its first site is occupied by the source. The prefactor comes from the
$2t-1$ identity inputs. In each term the sectors are fixed, but the
colors are still summed, with equal ket and bra colors at every
nonsource input.

We denote the output for these prescribed free words by
$\ket{\widehat\Phi_t((s_j)|(\tau_j))}$:
\begin{equation}
 \ket{\widehat\Phi_t((s_j)|(\tau_j))}
 =\widehat W_t\left[
 \left(\bigotimes_{j=1}^t\iota_{s_j}\right)\otimes
 \left(x_{O_\times}\otimes
       \bigotimes_{j=1}^{t-1}\iota_{\tau_j}\right)
 \right].
 \label{supp-eq:folded-fixed-word-state}
\end{equation}
This is the folded counterpart of Eq.~\eqref{supp-eq:fixed-word-state}.
The source is always inserted before the right word and is therefore
implicit in the arguments of $\widehat\Phi_t$. For example, the state
in Eq.~\eqref{main-eq:op-core-state} of the main text is
$\ket{\widehat\Phi_d(A^d|B^{d-1})}$ for $d\ge1$.

As in the domain-wall problem, we group the free words by their counts
$\ell=N_A((s_j))$ and $k=N_A((\tau_j))$. For
$0\le\ell\le t$ and $0\le k\le t-1$, the normalized count branch is
\begin{equation}
 \ket{\widehat\phi_{k\ell}}
 =\frac{1}{\sqrt{\binom{t}{\ell}\binom{t-1}{k}}}
 \sum_{\substack{(s_j)\in\{A,B\}^t\\N_A((s_j))=\ell}}
 \sum_{\substack{(\tau_j)\in\{A,B\}^{t-1}\\N_A((\tau_j))=k}}
 \ket{\widehat\Phi_t((s_j)|(\tau_j))}.
 \label{supp-eq:folded-fixed-count-word-sum}
\end{equation}
This is the analogue of Eq.~\eqref{supp-eq:fixed-count-word-sum}. Distinct
free words give orthogonal normalized inputs, and $\widehat W_t$
preserves their inner products. The sum therefore contains
$\binom{t}{\ell}\binom{t-1}{k}$ orthonormal terms, which fixes the
normalization.

In Secs.~\ref{supp-subsec:folded-reduction}
and~\ref{supp-subsec:folded-reconstruction} below, we work out the details
that differ from the domain-wall problem. Finally,
Sec.~\ref{supp-subsec:hermitian-source} shows that the Hermitian source
$O_{\rm H}=(O_\times+O_\times^\dagger)/\sqrt2$ has the same normalized
operator-Schmidt spectrum as $O_\times$.

\subsection{Common-word removal in the folded circuit}
\label{supp-subsec:folded-reduction}

We now adapt four earlier lemmas: fixed-word
reordering and coherent-branch reduction from
Sec.~\ref{supp-subsec:word-reordering}, word fixing from
Sec.~\ref{supp-subsec:word-fixing}, and common-word removal from
Sec.~\ref{supp-subsec:common-word-removal}. We discuss the required changes
in this order.

\head{Fixed-word reordering}
In Eq.~\eqref{supp-eq:folded-fixed-word-state}, the input permutations must
leave the source strand fixed. On every
nonsource pair, the input color state is
$2^{-1/2}\sum_c\ket{c}_{\rm k}\ket{c}_{\rm b}$.
A braid involving only these pairs applies the same color permutation
in ket and bra, so it preserves the uniform sum over equal color words.
This replaces the input-color argument in the original proof. The braid
relation then moves the input permutations through the folded rectangle
to separate unitaries on the output packets, preserving the Schmidt
spectrum.

\head{Coherent-branch reduction}
The state $\ket{\widehat\phi_{k\ell}}$ in
Eq.~\eqref{supp-eq:folded-fixed-count-word-sum} sums over distinct free words
with fixed counts, keeping $(B,A)$ at the source position.
The preceding reordering argument
provides the required unitaries. Their images for distinct words are
orthogonal because their pair-label sequences differ, so the normalized
sums again define separate isometries on the two output packets. With
this restriction on the words, the original proof applies without
further changes.

\head{Word fixing}
A permutation fixing a word of label pairs fixes the sector word in
each layer separately. The original lemma therefore makes its action
the identity in both layers, since adjacent equal-sector generators
are identities. No assumption on the color state is needed, so this
argument also applies to entangled ket and bra colors.

\head{Common-word removal}
Using the allowed input permutations, we place one copy of $u$ at the
beginning of the left packet and the other at the end of the right
packet. The source remains in the residual right word. As in the
original proof, we first apply the rectangle to the residual words;
the two copies of $u$ remain product factors across the cut. The
operations moving them past the residual words are local to the output
packets. The remaining rectangle exchanges $u|u$ and acts trivially by
the folded word-fixing argument. Thus deleting the two copies preserves
the nonzero Schmidt spectrum. The only change in the residual words is
that the distinguished $(B,A)$ pair cannot be removed, giving the two
cores in Eq.~\eqref{main-eq:op-cores} of the main text.

\subsection{Reconstruction of colors in the folded circuit}
\label{supp-subsec:folded-reconstruction}

Our goal is to determine the Schmidt spectrum of the count branch
$\ket{\widehat\phi_{k\ell}}$ defined in
Eq.~\eqref{supp-eq:folded-fixed-count-word-sum}. For $d=\ell-k\ge1$, the
reduction in Sec.~\ref{supp-subsec:folded-reduction} leaves the state
$\ket{\widehat\Phi_d(A^d|B^{d-1})}$, in the notation of
Eq.~\eqref{supp-eq:folded-fixed-word-state}. We will show that this state has
a flat spectrum with $2^{d-1}$ Schmidt coefficients. The case $d\le0$
will be treated below, after Eq.~\eqref{supp-eq:core-schmidt}, by exchanging
ket and bra.

The difference from Sec.~\ref{supp-subsec:color-reconstruction} is that both
input packets now carry color sums. Equation~\eqref{supp-eq:folded-sector-inputs}
shows that each nonsource input contributes a single color shared by
ket and bra. Write $a=(a_1,\ldots,a_d)\in\Ftwo^d$ for these colors in
the left packet and $b=(b_1,\ldots,b_{d-1})\in\Ftwo^{d-1}$ for those
on the right. Both source colors are zero, so the complete right input
color word is $0b=(0,b_1,\ldots,b_{d-1})$ in each layer. The fixed sector
words are $(A,A)^d\mid(B,A)(B,B)^{d-1}$; we suppress them in the color
kets below and list the complete ket word followed by the complete bra
word.

For fixed $(a,b)$, let $\mathcal L_{\rm k}(a,b)$ and
$\mathcal L_{\rm b}(a,b)$ denote the left output color words, and let
$\mathcal R_{\rm k}(a,b)$ and $\mathcal R_{\rm b}(a,b)$ denote the
right output words, all read from left to right within their packets.
The subscripts $\rm k$ and $\rm b$ denote the ket and bra color states,
respectively.
Expanding the diagonal color inputs in
Eq.~\eqref{supp-eq:folded-fixed-word-state} gives the folded counterpart of
Eq.~\eqref{supp-eq:bell-core-supp}:
\begin{equation}
 \begin{aligned}
 \ket{\widehat\Phi_d(A^d|B^{d-1})}
 &=\frac{1}{2^{d-1/2}}
   \sum_{a\in\Ftwo^d}\sum_{b\in\Ftwo^{d-1}}
   \widehat W_d\bigl(\ket{a,a}\otimes\ket{0b,0b}\bigr)\\
 &=\frac{1}{2^{d-1/2}}
   \sum_{a\in\Ftwo^d}\sum_{b\in\Ftwo^{d-1}}
   \ket{\mathcal L_{\rm k}(a,b),\mathcal L_{\rm b}(a,b)}
   \otimes
   \ket{\mathcal R_{\rm k}(a,b),\mathcal R_{\rm b}(a,b)}.
 \end{aligned}
 \label{supp-eq:folded-color-sum}
\end{equation}
The prefactor comes from the $2d-1$ normalized color sums. The sector
words differ at the source; therefore the colors in the ket and bra
layers need not be identical. However, we will show in the lemma for
the right output below that this output remains diagonal,
i.e.\ $\mathcal R_{\rm k}(a,b)=\mathcal R_{\rm b}(a,b)$. We then reconstruct $b$ from
that output and $(a,b)$ from the left output. These properties allow
us to collect the sum over $a$ into orthonormal left vectors and obtain
the Schmidt decomposition in Eq.~\eqref{supp-eq:core-schmidt} below.

To follow the color evolution, we label the incoming left strands by
$L_1,\ldots,L_d$ and the right strands by $R_0,\ldots,R_{d-1}$, in
left-to-right order within each packet. The strand $R_0$ is the source,
with sectors $(B,A)$; every $L_i$ has sectors $(A,A)$ and each remaining
$R_j$ has sectors $(B,B)$. Their input colors are $a_i$ on $L_i$, zero
on $R_0$, and $b_j$ on $R_j$ for $j\ge1$.
Each pair $L_i,R_j$ crosses once. We label this
crossing by $(i,j)$ and evaluate the circuit along
$R_0,\ldots,R_{d-1}$ in order; each crosses $L_d,\ldots,L_1$ and exits
on the left.

\begin{unnumberedlemma}[Reconstruction from the right output]
For $d\ge1$, the right output remains diagonal. Define
$\mathcal R(b)=\mathcal R_{\rm k}(a,b)=\mathcal R_{\rm b}(a,b)$.
The map $b\mapsto\mathcal R(b)$ is injective.
\end{unnumberedlemma}

\begin{proof}
Figure~\ref{supp-fig:folded-diagonal-rectangle} illustrates the diagonality
of the right output for $d=3$: all its outgoing $L_i$ segments are thick.
We first prove this for arbitrary $d$ and show that the common output
depends only on $b$, then reconstruct $b$ from it. We use the local relation
$h(a,c)=a+c$ derived in Sec.~\ref{supp-subsec:color-reconstruction}.

Let $u_i^{(j)}$ and $v_i^{(j)}$
be the ket and bra colors on $L_i$ after the first $j$ right strands have
crossed it. For $0\le j<d$, these are the colors immediately before the
crossing $(i,j)$ with $R_j$. The corresponding working packets are
$u^{(j)}=(u_1^{(j)},\ldots,u_d^{(j)})$ and
$v^{(j)}=(v_1^{(j)},\ldots,v_d^{(j)})$. Thus
$u^{(0)}=v^{(0)}=a$, while $u^{(d)}$ and $v^{(d)}$ are the final right
packets.

At every crossing with $R_0$, the ket layer has unequal sectors and uses
this relation, whereas
the bra gate has equal sectors and is the identity. Since both input
colors on $R_0$ are zero, its crossing with $L_d$ leaves
$u_d^{(1)}=v_d^{(1)}=0$. Suppose that the two working packets agree on
their last $q$ strands, with $1\le q<d$, before meeting some $R_j$ with
$j\ge1$. This
strand enters with the same color $b_j$ in ket and bra. Following $R_j$
through these $q$ crossings therefore gives equal moving colors in the
two layers. At the next crossing, with $L_{d-q}$, the unequal-sector gate copies
this common moving color onto that strand in both layers, independently
of its old color. The updated packets therefore agree on their last
$q+1$ strands. Starting from the common color on $L_d$ after its crossing with
$R_0$, the passages of $R_1,\ldots,R_{d-1}$ give
\begin{equation}
 u^{(d)}=v^{(d)}=\mathcal R(b).
 \label{supp-eq:common-right-packet}
\end{equation}
At each step the new common suffix depends only on the preceding common
suffix and the incoming color $b_j$. Hence the final word is independent
of $a$.

\begin{figure}[!tb]
\centering
\begin{tikzpicture}[x=1.65cm,y=.72cm,line cap=round,line join=round,
 font=\small,
 pairedwire/.style={draw=black!78,line width=.6pt},
 diagonalwire/.style={pairedwire,line width=1.8pt},
 pairedvertex/.style={circle,draw=blue!65!black,fill=white,
   minimum size=2.1mm,inner sep=0pt,line width=.7pt},
 pairlabel/.style={font=\small,inner sep=1pt},
 pairedstrand/.style={font=\footnotesize,text=black!70,inner sep=1pt}]
 \draw[diagonalwire] (0,0)--(0,2.30)--(.5,2.75);
 \draw[pairedwire] (.5,2.75)--(1.5,3.65);
 \draw[diagonalwire] (1.5,3.65)--(3,5.00)--(3,5.55);
 \draw[diagonalwire] (1,0)--(1,1.40)--(4,4.10)--(4,5.55);
 \draw[diagonalwire] (2,0)--(2,.50)--(5,3.20)--(5,5.55);
 \draw[diagonalwire] (3,0)--(3,.50)--(1.5,1.85);
 \draw[pairedwire] (1.5,1.85)--(0,3.20)--(0,5.55);
 \draw[diagonalwire] (4,0)--(4,1.40)--(1.5,3.65);
 \draw[pairedwire] (1.5,3.65)--(1,4.10)--(1,5.55);
 \draw[diagonalwire] (5,0)--(5,2.30)--(2,5.00)--(2,5.55);
 \foreach \x/\y in {2.5/.95,1.5/1.85,3.5/1.85,
   .5/2.75,2.5/2.75,4.5/2.75,1.5/3.65,3.5/3.65,2.5/4.55}
   \node[pairedvertex] at (\x,\y){};

 \foreach \x/\i in {0/1,1/2,2/3}{
   \node[pairlabel,anchor=north] at (\x,-.08) {$(a_{\i},a_{\i})$};
   \node[pairedstrand,anchor=north] at (\x,-.80) {$L_{\i}$};
 }
 \node[pairlabel,anchor=north,draw=blue!65!black,rounded corners=1pt,
   inner sep=2pt] at (3,-.08) {$(0,0)$};
 \node[pairedstrand,anchor=north] at (3,-.80) {$R_0$};
 \node[pairlabel,anchor=north] at (3,-1.40) {source $O_\times$};
 \foreach \x/\i in {4/1,5/2}{
   \node[pairlabel,anchor=north] at (\x,-.08) {$(b_{\i},b_{\i})$};
   \node[pairedstrand,anchor=north] at (\x,-.80) {$R_{\i}$};
 }
 \foreach \x/\j in {0/0,1/1,2/2}
   \node[pairedstrand,anchor=south] at (\x,5.65) {$R_{\j}$};
 \foreach \x/\i in {3/1,4/2,5/3}
   \node[pairedstrand,anchor=south] at (\x,5.65) {$L_{\i}$};
 \node[pairlabel,anchor=south] at (1,6.20) {left output};
 \node[pairlabel,anchor=south] at (4,6.20) {right output};
 \draw[dashed,black!45] (2.5,5.15)--(2.5,6.65);
\end{tikzpicture}
\par\smallskip
\begin{tikzpicture}[font=\small,line cap=round]
 \draw[black!78,line width=1.8pt] (0,0)--(.60,0);
 \node[anchor=west] at (.72,0) {equal ket and bra colors};
 \draw[black!78,line width=.6pt] (5.40,0)--(6,0);
 \node[anchor=west] at (6.12,0) {colors may differ};
\end{tikzpicture}
\caption{Folded version of the rectangle in
Fig.~\ref{supp-fig:state-color-rectangle}(a), for $d=3$. Each line carries one
ket--bra color pair. Thick segments have equal colors for every input
word $(a,b)$; thin segments need not. The displayed input pairs are
diagonal in color. The boxed source $R_0$ carries sectors $(B,A)$,
whereas the other inputs carry $(A,A)$ on the left and $(B,B)$ on the
right. The different sector labels on $R_0$ lead to different gates
acting on the colors in the two layers: at its crossings with the
$L_i$, the ket layer applies the unequal-sector color rule, while the
bra gate is the identity. This can generate non-diagonal color pairs,
whose ket--bra differences propagate through subsequent crossings.
However, a crossing with $R_0$ does not necessarily break diagonality.
For example, $R_0$ enters its crossing with $L_3$ with colors $(0,0)$.
Since $h(a_3,0)=a_3$, the outgoing pairs are $(0,0)$ on $L_3$ and
$(a_3,a_3)$ on $R_0$, so both remain diagonal for every $a_3$.
All final segments on the $L_i$ strands at the right output are thick,
in agreement with
Eq.~\eqref{supp-eq:common-right-packet}.}
\label{supp-fig:folded-diagonal-rectangle}
\end{figure}
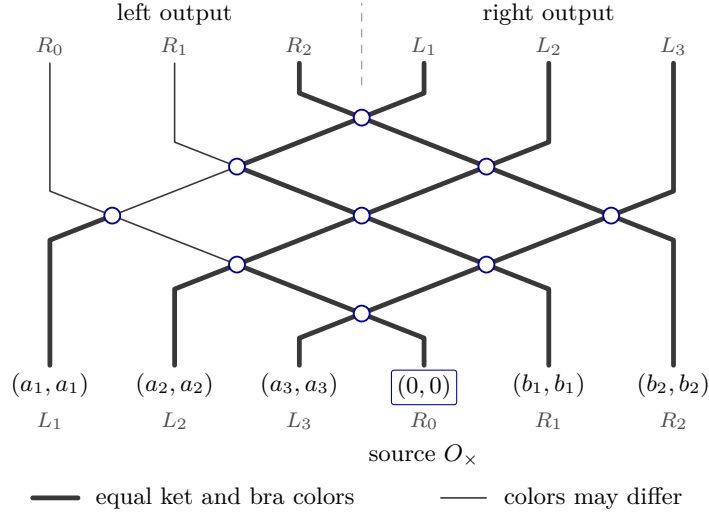

For $d=1$ there is no $b$ variable, so the injectivity of $\mathcal R$ is
immediate. For $d\ge2$, we use the same backwards reconstruction as in
Fig.~\ref{supp-fig:state-color-reconstruction}(a). This reconstruction does
not require the initial left colors $a$; here the right input sequence
is $(0,b_1,\ldots,b_{d-1})$, with the source color fixed to zero.
Starting from $\mathcal R(b)$, follow $R_{d-1},\ldots,R_1$ in reverse
order. Before undoing the crossings with each $R_j$, the last entry of the known
packet suffix is $b_j$. The middle relation in
Eq.~\eqref{supp-eq:local-peeling} reconstructs the colors before those
crossings, with the first strand of the known suffix omitted. For $j>1$,
the last color of this shorter suffix is $b_{j-1}$, which starts the next
step. This recovers $b_{d-1},b_{d-2},\ldots,b_1$ and proves the
injectivity of $\mathcal R$.
\end{proof}

The left output must distinguish the complete input record $(a,b)$.
This is the second ingredient needed for the Schmidt decomposition.

\begin{unnumberedlemma}[Reconstruction from the left output]
For $d\ge1$, the joint map
$(a,b)\mapsto(\mathcal L_{\rm k}(a,b),\mathcal L_{\rm b}(a,b))$
is injective.
\end{unnumberedlemma}

\begin{proof}
We retain the notation $u_i^{(j)}$ and $v_i^{(j)}$ from the preceding
proof. Starting from the first input color, we reconstruct the ket
colors along one left strand, use the already proved
Eq.~\eqref{supp-eq:common-right-packet} at its outgoing end, and return along
the bra strand to recover the next input color.
For $d=3$, Fig.~\ref{supp-fig:folded-color-reconstruction} shows the
alternating ket sweeps and bra returns described below.
Let $y_j^{\rm k}$ and $y_j^{\rm b}$ be the outgoing ket and bra
colors on $R_j$; these are precisely the entries of $\mathcal L_{\rm k}$
and $\mathcal L_{\rm b}$. Every bra crossing with $R_0$ is the identity in site order, whereas
our variables follow the labeled strands through the braid. At each
crossing, the outgoing $L_i$ therefore takes the color carried by $R_0$,
and $R_0$ takes the old color on $L_i$. Following $R_0$ from $L_d$ to
$L_1$ shifts the recorded colors by one place and gives
\begin{equation}
 y_0^{\rm b}=a_1,\qquad
 v_i^{(1)}=a_{i+1}\ (1\le i<d),\qquad
 v_d^{(1)}=0.
 \label{supp-eq:bra-source-boundary}
\end{equation}
We therefore know $a_1$. For $d=1$ this already proves joint injectivity,
because there is no $b$ variable. We assume $d\ge2$ in the remainder of the
reconstruction. The ket crossing $(1,0)$ gives
$y_0^{\rm k}=h(a_1,u_1^{(1)})$ and determines $u_1^{(1)}$. The later
crossings $(1,j)$ give
\begin{equation}
 y_j^{\rm k}=h(u_1^{(j)},u_1^{(j+1)}),\qquad 1\le j<d,
 \label{supp-eq:ket-forward-boundary}
\end{equation}
so following $L_1$ through its ket crossings determines
$u_1^{(2)},\ldots,u_1^{(d)}$. Equation~\eqref{supp-eq:common-right-packet}
then supplies $v_1^{(d)}=u_1^{(d)}$. The corresponding bra relations
$y_j^{\rm b}=h(v_1^{(j)},v_1^{(j+1)})$, now read from
$j=d-1$ down to $j=1$ along the same strand, determine
$v_1^{(d-1)},\ldots,v_1^{(1)}$. The last of these values is $a_2$ by
Eq.~\eqref{supp-eq:bra-source-boundary}.

We now repeat the construction along $L_2,\ldots,L_{d-1}$. Suppose that
the ket and bra colors after every crossing on $L_1,\ldots,L_{i-1}$,
together with the input color $a_i$, are already known. The ket crossing
$(i,0)$ gives
\begin{equation}
 u_{i-1}^{(1)}=h(a_i,u_i^{(1)})
 \label{supp-eq:ket-source-strand}
\end{equation}
and determines $u_i^{(1)}$. At each subsequent crossing $(i,j)$,
$1\le j<d$, the local relation
\begin{equation}
 u_{i-1}^{(j+1)}=h(u_i^{(j)},u_i^{(j+1)})
 \label{supp-eq:ket-forward-strand}
\end{equation}
determines the next ket color on $L_i$. After reaching its outgoing end,
equality of the final packets supplies $v_i^{(d)}=u_i^{(d)}$. We then use
\begin{equation}
 v_{i-1}^{(j+1)}=h(v_i^{(j)},v_i^{(j+1)})
 \label{supp-eq:bra-backward-strand}
\end{equation}
in decreasing order $j=d-1,\ldots,1$ to reconstruct the bra colors
backwards along $L_i$. The last recovered color is
$v_i^{(1)}=a_{i+1}$. Repeating these forward ket and backward bra steps
along $L_2,\ldots,L_{d-1}$ recovers $a_3,\ldots,a_d$ and hence the
complete word $a$; for $d=2$, the reconstruction along $L_1$ already
suffices.

Finally, with $a$ known, follow the source $R_0$ through its ket crossings
to obtain $u^{(1)}$. For each $R_j$, $j=1,\ldots,d-1$, the packet
$u^{(j)}$ and the outgoing color $y_j^{\rm k}$ are now known.
Starting from this outgoing color, read the crossings backwards along
$R_j$, from $L_1$ to $L_d$. Local invertibility determines the successive
ket colors and ends at the incoming color $b_j$. The same reconstruction
gives $u^{(j+1)}$, so we can proceed to $R_{j+1}$ when $j<d-1$. This
recovers $b_1,\ldots,b_{d-1}$ and completes the proof of joint
injectivity.
\end{proof}

We now construct the Schmidt decomposition by collecting the sum over
$a$ in Eq.~\eqref{supp-eq:folded-color-sum}. Define the left vectors by
\begin{equation}
 \ket{\lambda_b}=2^{-d/2}\sum_{a\in\Ftwo^d}
 \ket{\mathcal L_{\rm k}(a,b),\mathcal L_{\rm b}(a,b)}.
\end{equation}
Joint injectivity makes these vectors normalized, with disjoint supports
for distinct $b$. The normalized core state is
\begin{equation}
 \ket{\widehat\Phi_d(A^d|B^{d-1})}
 =2^{-(d-1)/2}\sum_{b\in\Ftwo^{d-1}}
 \ket{\lambda_b}\otimes\ket{\mathcal R(b),\mathcal R(b)}.
 \label{supp-eq:core-schmidt}
\end{equation}
The left vectors are orthonormal by joint injectivity, and the right
vectors by injectivity of $\mathcal R$. Thus Eq.~\eqref{supp-eq:core-schmidt}
is a flat Schmidt decomposition with $2^{d-1}$ equal coefficients.

For $d\le0$, put $n=1-d$ and use the same input-order labeling with $n$
strands in each packet. Now $L_1,\ldots,L_n$ have sectors $(B,B)$,
$R_0$ is still the source $(B,A)$, and $R_1,\ldots,R_{n-1}$ have sectors
$(A,A)$. At every crossing with $R_0$, the ket gate has equal sectors
and the bra gate is mixed. Thus the same color reconstruction applies
with ket and bra exchanged and gives $n-1=-d$ Schmidt bits (each Schmidt bit
denotes a contribution of $\log2$ to the entanglement, the value that would
be carried by a single Bell pair of qubits).
Thus this core also has a flat spectrum, with rank $2^{-d}$.

For both cases, the Schmidt probabilities within each normalized count
branch are
\begin{equation}
 \mu_j^{(d)}=2^{-g(d)}\quad(1\le j\le2^{g(d)}),\qquad
 g(d)=\left|d-\tfrac12\right|-\tfrac12,\qquad d=\ell-k.
 \label{supp-eq:op-core}
\end{equation}
The branch entropy is therefore $(\log2)g(d)$.
The finite correction to $|d|\log2$ comes from crossings with the
source: one layer undergoes a mixed-sector shear, while the other sees
an equal-sector identity.

\begingroup
\tikzset{
 f-wire/.style={draw=black!32,line width=.55pt},
 f-vertex/.style={circle,draw=black!65,fill=white,
   minimum size=2.2mm,inner sep=0pt,line width=.65pt},
 f-source/.style={f-vertex,rectangle,minimum size=2.4mm},
 f-active/.style={line width=1.75pt,postaction={decorate},
   decoration={markings,mark=at position .54 with
     {\arrow{Stealth[length=2.8mm,width=2.1mm]}}}},
 f-feed/.style={line width=1.05pt,postaction={decorate},
   decoration={markings,mark=at position .60 with
     {\arrow{Stealth[length=2mm,width=1.5mm]}}}},
 f-history/.style={line width=1.3pt},
 f-label/.style={font=\small,inner sep=1pt},
 f-strand/.style={font=\scriptsize,text=black!65,inner sep=1pt}
}

\newcommand{\FoldedWires}{%
 \foreach \x in {0,...,5}\draw[f-wire](\x,0)--(\x,.50);
 \foreach \x in {0,1,4,5}\draw[f-wire](\x,.50)--(\x,1.40);
 \foreach \x/\xp in {2/3,3/2}\draw[f-wire](\x,.50)--(\xp,1.40);
 \foreach \x in {0,5}\draw[f-wire](\x,1.40)--(\x,2.30);
 \foreach \x/\xp in {1/2,2/1,3/4,4/3}
   \draw[f-wire](\x,1.40)--(\xp,2.30);
 \foreach \x/\xp in {0/1,1/0,2/3,3/2,4/5,5/4}
   \draw[f-wire](\x,2.30)--(\xp,3.20);
 \foreach \x in {0,5}\draw[f-wire](\x,3.20)--(\x,4.10);
 \foreach \x/\xp in {1/2,2/1,3/4,4/3}
   \draw[f-wire](\x,3.20)--(\xp,4.10);
 \foreach \x in {0,1,4,5}\draw[f-wire](\x,4.10)--(\x,5.00);
 \foreach \x/\xp in {2/3,3/2}\draw[f-wire](\x,4.10)--(\xp,5.00);
 \foreach \x in {0,...,5}\draw[f-wire](\x,5.00)--(\x,5.55);
}
\newcommand{\FoldedVertices}[1]{%
 \foreach \x/\y in {3.5/1.85,2.5/2.75,4.5/2.75,1.5/3.65,3.5/3.65,2.5/4.55}
   \node[f-vertex] at (\x,\y){};
 \ifnum#1=0
   \foreach \x/\y in {2.5/.95,1.5/1.85,.5/2.75}
     \node[f-vertex] at (\x,\y){};
 \else
   \foreach \x/\y in {2.5/.95,1.5/1.85,.5/2.75}
     \node[f-source] at (\x,\y){};
 \fi
}
\newcommand{\FoldedLabels}[1]{%
 \foreach \x/\i in {0/1,1/2,2/3}{
   \node[f-label,anchor=north] at (\x,-.10) {$a_{\i}$};
   \node[f-strand,anchor=north] at (\x,-.56) {$L_{\i}$};
 }
 \node[f-label,anchor=north] at (3,-.10) {$0$};
 \node[f-strand,anchor=north] at (3,-.56) {$R_0$};
 \foreach \x/\i in {4/1,5/2}{
   \node[f-label,anchor=north] at (\x,-.10) {$b_{\i}$};
   \node[f-strand,anchor=north] at (\x,-.56) {$R_{\i}$};
 }
 \ifnum#1=0
   \foreach \x/\j in {0/0,1/1,2/2}
     \node[f-label,anchor=south] at (\x,5.68) {$y_{\j}^{\rm k}$};
   \foreach \x/\i in {3/1,4/2,5/3}
     \node[f-label,anchor=south] at (\x,5.68) {$u_{\i}^{(3)}$};
 \else
   \foreach \x/\j in {0/0,1/1,2/2}
     \node[f-label,anchor=south] at (\x,5.68) {$y_{\j}^{\rm b}$};
   \foreach \x/\i in {3/1,4/2,5/3}
     \node[f-label,anchor=south] at (\x,5.68) {$v_{\i}^{(3)}$};
 \fi
}
\newcommand{\FoldedPanel}[3][]{%
 \begin{tikzpicture}[x=1.18cm,y=.78cm,line cap=round,line join=round,#1]
 \path[use as bounding box](-.48,-.98) rectangle (5.48,6.45);
 \FoldedWires
 #3
 \FoldedVertices{#2}
 \FoldedLabels{#2}
 \end{tikzpicture}%
}
\newcommand{\FoldedKetOne}[1]{%
 \draw[#1] (0,0)--(0,2.30)--(.5,2.75);
 \draw[#1] (.5,2.75)--(1.5,3.65);
 \draw[#1] (1.5,3.65)--(2.5,4.55);
 \draw[#1] (2.5,4.55)--(3,5.00)--(3,5.55);
}
\newcommand{\FoldedKetTwo}[1]{%
 \draw[#1] (1,0)--(1,1.40)--(1.5,1.85);
 \draw[#1] (1.5,1.85)--(2.5,2.75);
 \draw[#1] (2.5,2.75)--(3.5,3.65);
 \draw[#1] (3.5,3.65)--(4,4.10)--(4,5.55);
}
\newcommand{\FoldedBraSeed}[1]{%
 \draw[#1] (0,5.55)--(0,3.20)--(.5,2.75);
 \draw[#1] (.5,2.75)--(0,2.30)--(0,0);
}
\newcommand{\FoldedBraOne}[1]{%
 \draw[#1] (3,5.55)--(3,5.00)--(2.5,4.55);
 \draw[#1] (2.5,4.55)--(1.5,3.65);
 \draw[#1] (1.5,3.65)--(.5,2.75);
 \draw[#1] (.5,2.75)--(1.5,1.85);
 \draw[#1] (1.5,1.85)--(1,1.40)--(1,0);
}
\newcommand{\FoldedBraTwo}[1]{%
 \draw[#1] (4,5.55)--(4,4.10)--(3.5,3.65);
 \draw[#1] (3.5,3.65)--(2.5,2.75);
 \draw[#1] (2.5,2.75)--(1.5,1.85);
 \draw[#1] (1.5,1.85)--(2.5,.95);
 \draw[#1] (2.5,.95)--(2,.50)--(2,0);
}
\newcommand{\FoldedKetFeedOne}[1]{%
 \draw[#1] (0,5.55)--(0,3.20)--(.5,2.75);
 \draw[#1] (1,5.55)--(1,4.10)--(1.5,3.65);
 \draw[#1] (2,5.55)--(2,5.00)--(2.5,4.55);
}
\newcommand{\FoldedBraFeedOne}[1]{%
 \draw[#1] (1,5.55)--(1,4.10)--(1.5,3.65);
 \draw[#1] (2,5.55)--(2,5.00)--(2.5,4.55);
}
\newcommand{\FoldedKetFeedTwo}[1]{%
 \draw[#1] (.5,2.75)--(1.5,1.85);
 \draw[#1] (1.5,3.65)--(2.5,2.75);
 \draw[#1] (2.5,4.55)--(3.5,3.65);
}
\newcommand{\FoldedBraFeedTwo}[1]{%
 \draw[#1] (1.5,3.65)--(2.5,2.75);
 \draw[#1] (2.5,4.55)--(3.5,3.65);
}

\begin{figure}[!t]
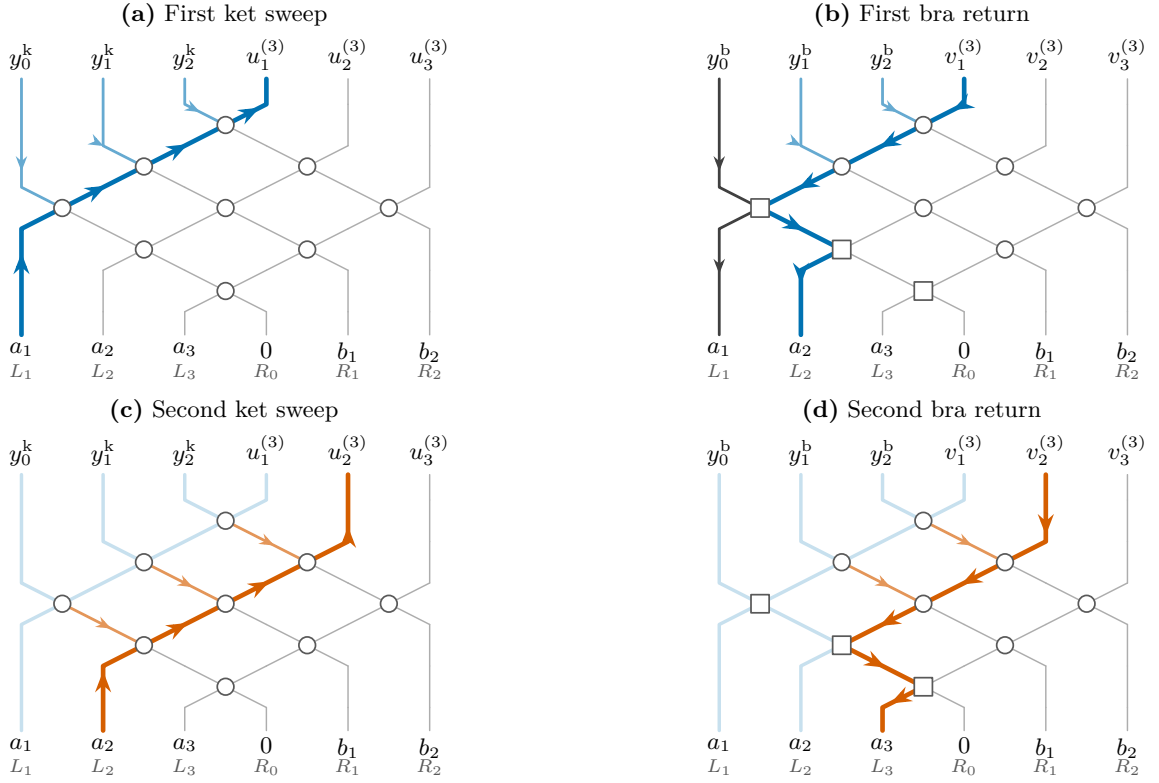

\centering
\begin{minipage}[t]{.485\linewidth}\centering
\textbf{(a)} First ket sweep\par\smallskip
\FoldedPanel[x=1.08cm,y=.61cm]{0}{%
\FoldedKetFeedOne{f-feed,reconfirst!60}
\FoldedKetOne{f-active,reconfirst}
}
\end{minipage}\hfill
\begin{minipage}[t]{.485\linewidth}\centering
\textbf{(b)} First bra return\par\smallskip
\FoldedPanel[x=1.08cm,y=.61cm]{1}{%
\FoldedBraSeed{f-feed,black!75}
\FoldedBraFeedOne{f-feed,reconfirst!60}
\FoldedBraOne{f-active,reconfirst}
}
\end{minipage}
\par\medskip
\begin{minipage}[t]{.485\linewidth}\centering
\textbf{(c)} Second ket sweep\par\smallskip
\FoldedPanel[x=1.08cm,y=.61cm]{0}{%
\FoldedKetOne{f-history,reconfirst!22}
\FoldedKetFeedOne{f-history,reconfirst!22}
\FoldedKetFeedTwo{f-feed,reconsecond!65}
\FoldedKetTwo{f-active,reconsecond}
}
\end{minipage}\hfill
\begin{minipage}[t]{.485\linewidth}\centering
\textbf{(d)} Second bra return\par\smallskip
\FoldedPanel[x=1.08cm,y=.61cm]{1}{%
\FoldedBraSeed{f-history,reconfirst!22}
\FoldedBraOne{f-history,reconfirst!22}
\FoldedBraFeedOne{f-history,reconfirst!22}
\FoldedBraFeedTwo{f-feed,reconsecond!65}
\FoldedBraTwo{f-active,reconsecond}
}
\end{minipage}
\caption{Left reconstruction for $d=3$. Panels (a,c) show the ket layer,
and panels (b,d) show the bra layer. In the text we introduced
$u_i^{(j)}$ and $v_i^{(j)}$ as the intermediate ket and bra colors on
the left strand $L_i$. Here their values at $j=3$ label the outgoing
endpoints.
The lemma \emph{Reconstruction from the right output} establishes
$u_i^{(3)}=v_i^{(3)}$ for $i=1,2,3$.
Squares mark identity gates where both crossing strands carry sector
$A$; in this core, these occur only on $R_0$ in the bra layer.
The reconstruction starts with $a_1=y_0^{\rm b}$, a step not shown
separately: the identity crossing of $R_0$ with $L_1$ places the input
color $a_1$ at the left output. With this seed, panels (a,b) show the
sweep along $L_1$ and the bra return through the identity crossings
to the input of $L_2$; panels (c,d) repeat this step along $L_2$,
ending at the input of $L_3$.
Saturated arrows show the current sweep,
pale blue lines retain the preceding round trip, and lighter arrows
supply known colors.}
\label{supp-fig:folded-color-reconstruction}
\end{figure}

\endgroup

\subsection{A Hermitian source}
\label{supp-subsec:hermitian-source}

We now show that replacing $O_\times$ by $O_{\rm H}$ preserves the
normalized Schmidt spectrum of the full state $\ket{\Omega_t}$ in
Eq.~\eqref{supp-eq:folded-rectangle-supp}. We need right-output diagonality
before count projection. The induction used in the right-output lemma
of Sec.~\ref{supp-subsec:folded-reconstruction} applies to every fixed-word
state in Eq.~\eqref{supp-eq:folded-fixed-word-state}: the source gives the
last left strand equal ket and bra colors, and each remaining right
strand extends the common suffix by one. This works for arbitrary free
sector words because, for both equal and unequal sectors, the local
rule copies the incoming right-strand color onto the outgoing left
strand. The right output also has equal ket and bra sectors, inherited
from the left inputs. Hence the entire coherent sum in
Eq.~\eqref{supp-eq:folded-sector-expansion} has diagonal right support.

Let $M$ be the normalized coefficient matrix of $\ket{\Omega_t}$ for
$O_\times$, with left output states indexing rows and right output
states indexing columns. Its entries are real, so taking the adjoint
of the source exchanges ket and bra. This exchange acts trivially on
the diagonal right support; the coefficient matrix for
$O_\times^\dagger$ is therefore $VM$, where $V$ exchanges ket and bra
on the left. The unitary $V$ changes the unique source pair from
$(B,A)$ to $(A,B)$, giving orthogonal left supports. Thus
$V^\dagger V=\id$ and $M^\dagger VM=0$.

For the Hermitian source, the normalized coefficient matrix is
$M_{\rm H}=(M+VM)/\sqrt2$. Both cross terms vanish by
$M^\dagger VM=0$ and its adjoint, $M^\dagger V^\dagger M=0$.
Using $V^\dagger V=\id$, we obtain
\begin{equation*}
 M_{\rm H}^\dagger M_{\rm H}
 =\frac12\bigl(M^\dagger M+M^\dagger V^\dagger VM\bigr)
 =M^\dagger M.
\end{equation*}
The eigenvalues of these Gram matrices are the squared Schmidt
coefficients, so the two sources have the same normalized Schmidt
spectrum.

\section{Entropy bounds and square-root growth}
\label{supp-sec:coherent-count}

In this section we prove a general entropy bound for superpositions of
orthogonal count sectors. We apply it to the domain-wall and folded
operator states and then evaluate the two binomial averages that
determine the leading von Neumann entropies.

\subsection{Entropy bound}

To treat both problems in one statement, consider a normalized
bipartite pure state decomposed into orthogonal count sectors,
\begin{equation}
 \ket{\Psi}=\sum_{k,\ell}\sqrt{\pi_{k\ell}}\ket{\chi_{k\ell}}
 \label{supp-eq:coherent-count-state}
\end{equation}
where $\ket{\chi_{k\ell}}$ is the normalized branch with count labels
$k,\ell$. We write $\Prob$ for probability, so that
$\pi_{k\ell}=\Prob(K=k,L=\ell)$ is the joint distribution of the count
variables $K,L$. The two counts are locally resolved by
orthogonal left and right projectors $P_k^{\rm L}$ and $P_\ell^{\rm R}$
satisfying
\begin{equation}
 (P_k^{\rm L}\otimes P_\ell^{\rm R})\ket{\chi_{k'\ell'}}
 =\delta_{kk'}\delta_{\ell\ell'}\ket{\chi_{k'\ell'}}.
 \label{supp-eq:count-sector-support}
\end{equation}
For any bipartite pure state $\ket{\psi}$, we write
$S_1(\psi)=S(\operatorname{tr}_{\rm L}\ket{\psi}\bra{\psi})$, where
$S(\rho)=-\operatorname{tr}(\rho\log\rho)$, and we denote the Shannon
entropy of the joint count distribution by
$H(K,L)=-\sum_{k,\ell}\pi_{k\ell}\log \pi_{k\ell}$.

\begin{unnumberedlemma}[Count-sector entropy bound]
We have the bounds
\begin{equation}
 \sum_{k,\ell}\pi_{k\ell}S_1(\chi_{k\ell})
 \le S_1(\Psi)
 \le H(K,L)+\sum_{k,\ell}\pi_{k\ell}S_1(\chi_{k\ell}).
 \label{supp-eq:general-coherent-count-bound}
\end{equation}
Thus the entropy of the full state is at least the mean branch
entropy and exceeds it by at most the Shannon entropy of the count
labels.
\end{unnumberedlemma}

\begin{proof}
We first group the branches by the left count. Define the conditional
branch states by
\begin{equation}
 \pi_k=\sum_\ell \pi_{k\ell},\qquad
 \pi_{\ell|k}=\frac{\pi_{k\ell}}{\pi_k},\qquad
 \ket{\chi_k}=\sum_\ell\sqrt{\pi_{\ell|k}}\ket{\chi_{k\ell}},
 \label{supp-eq:conditional-count-state}
\end{equation}
where values of $k$ with $\pi_k=0$ can be omitted. Let $\rho_{\rm R}$ be
the right reduced density matrix of $\ket{\Psi}$ and let
$\rho_{\rm R}^{(k)}$ be that of $\ket{\chi_k}$. Similarly, let
$\rho_{\rm L}^{(k)}$ and $\rho_{\rm L}^{(k\ell)}$ be the left reduced
density matrices of $\ket{\chi_k}$ and $\ket{\chi_{k\ell}}$.
Tracing over the orthogonal left count sectors in
Eq.~\eqref{supp-eq:count-sector-support} removes cross terms with different
$k$ and gives the first decomposition below. Within a fixed $k$ branch,
tracing over the orthogonal right count sectors removes cross terms with
different $\ell$ and gives the second:
\begin{equation}
 \rho_{\rm R}=\sum_k \pi_k\rho_{\rm R}^{(k)},\qquad
 \rho_{\rm L}^{(k)}
 =\sum_\ell \pi_{\ell|k}\rho_{\rm L}^{(k\ell)}.
 \label{supp-eq:count-density-decompositions}
\end{equation}
The two reduced density matrices of a pure state have the same nonzero
spectrum. Therefore
$S(\rho_{\rm R}^{(k)})=S(\rho_{\rm L}^{(k)})$.
We now use the standard entropy-of-a-mixture bounds
\cite{supp-NielsenChuang2010},
\begin{equation}
 \sum_i r_iS(\sigma_i)
 \le S\left(\sum_i r_i\sigma_i\right)
 \le H(\{r_i\})+\sum_i r_iS(\sigma_i).
 \label{supp-eq:entropy-of-mixture}
\end{equation}
Applying first the lower inequality to both decompositions in
Eq.~\eqref{supp-eq:count-density-decompositions} gives
\begin{equation}
 S_1(\Psi)\ge
 \sum_{k,\ell}\pi_{k\ell}S_1(\chi_{k\ell}).
\end{equation}
Applying instead the upper inequality gives
\begin{align}
 S_1(\Psi)
 &\le H(K)+\sum_k\pi_k H(L|K=k)\nonumber\\
 &\hspace{1.2cm}+\sum_{k,\ell}\pi_{k\ell}S_1(\chi_{k\ell})\nonumber\\
 &=H(K,L)+\sum_{k,\ell}\pi_{k\ell}S_1(\chi_{k\ell}).
\end{align}
This proves the lemma.
\end{proof}

Independence is used only when evaluating $H(K,L)$ for the binomial
counts in Sec.~\ref{supp-subsec:binomial-vn} below.
For the domain-wall state, setting
$\ket{\Psi}=\ket{\Phi_t}$, $\ket{\chi_{k\ell}}=\ket{\phi_{k\ell}}$,
and $\pi_{k\ell}=\pi^{\dw}_{k\ell}$ gives
Eq.~\eqref{main-eq:coherent-count-bound} of the main text.
For the operator problem, setting
$\ket{\Psi}=\ket{\Omega_t}$,
$\ket{\chi_{k\ell}}=\ket{\widehat\phi_{k\ell}}$, and taking
$\pi_{k\ell}$ to be the product binomial distribution given in
Eq.~\eqref{supp-eq:count-weights} below gives the corresponding bound for
the operator entropy.

\subsection{Binomial averages and von Neumann entropy}
\label{supp-subsec:binomial-vn}

We assume throughout this subsection that $t>1$. We now compute the mean
branch entropies and give the corresponding finite-time bounds. The state
and operator problems involve
centered binomial variables with $2t$ and $2t-1$ trials, respectively.
Their mean absolute deviations are equal; the inserted local operator
shifts the mean operator branch entropy by a constant. Stirling's formula
then gives the two square-root laws.

\head{Domain-wall state}
The independent counts are $K,L\sim\Bin(t,1/2)$, where $\Bin(n,p)$
denotes the binomial distribution with $n$ trials and success probability
$p$. A normalized branch has entropy $|K-L|\log2$. We denote the mean
branch entropy by
$\overline S^{\dw}(t)=(\log2)\E|K-L|$.
Since $t-L$ has the same binomial distribution as $L$, the difference
$K-L$ has the law of $Z-t$, where $Z\sim\Bin(2t,1/2)$.
Symmetry about $t$ and a telescoping sum give
\begin{align}
 \E|K-L|
 &=\frac{2}{4^t}\sum_{j=t+1}^{2t}(j-t)\binom{2t}{j}\nonumber\\
 &=\frac{2t}{4^t}\sum_{j=t+1}^{2t}
 \left[\binom{2t-1}{j-1}-\binom{2t-1}{j}\right]
 =t\frac{\binom{2t}{t}}{4^t}.
 \label{supp-eq:state-binomial-mean}
\end{align}
Thus the average in Eq.~\eqref{main-eq:state-mean} of the main text is
exactly $\overline S^{\dw}(t)=(\log2)t\binom{2t}{t}/4^t$.
Independence gives $H(K,L)=2H(\Bin(t,1/2))$, where
$H(\Bin(n,1/2))$ denotes the Shannon entropy of a binomial distribution.
Applying Eq.~\eqref{supp-eq:general-coherent-count-bound}, we obtain
\begin{equation}
 \overline S^{\dw}(t)\le S_1^{\dw}(t)
 \le\overline S^{\dw}(t)+2H(\Bin(t,1/2)).
 \label{supp-eq:state-vn-finite}
\end{equation}

\head{Operator}
Now $K\sim\Bin(t-1,1/2)$ and $L\sim\Bin(t,1/2)$ are independent.
Since $t-1-K$ has the same binomial distribution as $K$, the variable
$\Sigma=L-K+t-1$ has distribution $\Bin(2t-1,1/2)$.
The branch entropy from Eq.~\eqref{supp-eq:op-core} is $(\log2)g(L-K)$, with
$g(L-K)=|\Sigma-(2t-1)/2|-1/2$.
The same telescoping step gives
\begin{align}
 \E\left|\Sigma-\frac{2t-1}{2}\right|
 &=\frac{1}{2^{2t-2}}\sum_{j=t}^{2t-1}
 \left(j-\frac{2t-1}{2}\right)\binom{2t-1}{j}\nonumber\\
 &=\frac{2t-1}{2^{2t-1}}\binom{2t-2}{t-1}
 =t\frac{\binom{2t}{t}}{4^t}.
 \label{supp-eq:operator-binomial-mean}
\end{align}
The mean operator branch entropy is therefore
\begin{equation}
 \overline S^{\op}(t)=(\log2)\E g(L-K)
 =\overline S^{\dw}(t)-\tfrac12\log2.
 \label{supp-eq:operator-mean-offset}
\end{equation}
This exact offset concerns the mean branch entropies; the entropies of
the full states are bounded using
Eq.~\eqref{supp-eq:general-coherent-count-bound}. The uniform bound
$0\le|d|-g(d)\le1$ justifies absorbing the correction into an $O(1)$
remainder in Eq.~\eqref{main-eq:op-mean} of the main text.
The count-sector bound gives
\begin{equation}
 \overline S^{\op}(t)\le S_1^{\op}(t)
 \le\overline S^{\op}(t)+H(\Bin(t-1,1/2))+H(\Bin(t,1/2)).
 \label{supp-eq:operator-vn-finite}
\end{equation}

\head{Asymptotics}
Stirling's formula yields
\begin{equation}
 t\frac{\binom{2t}{t}}{4^t}=\sqrt{\frac{t}{\pi}}+O(t^{-1/2}).
 \label{supp-eq:binomial-mean-asymptotic}
\end{equation}
Multiplying this result by $\log2$ gives the state mean entropy; subtracting
$\tfrac12\log2$ gives the operator mean entropy. The count pairs have
at most $(t+1)^2$ possible values in the state problem and $t(t+1)$ in
the operator problem. Their Shannon entropies are therefore $O(\log t)$.
The finite-time bounds~\eqref{supp-eq:state-vn-finite} and
\eqref{supp-eq:operator-vn-finite} consequently prove
Eqs.~\eqref{main-eq:state-vn} and \eqref{main-eq:op-vn} of the main text.
More precisely, the binomial entropy asymptotic
$H(\Bin(n,1/2))=\tfrac12\log n+O(1)$ shows that both upper-bound
corrections are $\log t+O(1)$.

\section{R\'enyi entropies}
\label{supp-sec:renyi}

We derive the R\'enyi bounds stated in
Eq.~\eqref{main-eq:hierarchy} of the main text. The input is the flat
spectrum of each count branch, proved in Sec.~\ref{supp-subsec:folded-reconstruction}
and summarized in Eq.~\eqref{supp-eq:op-core}.
The remaining task is to control the coherent sum of these branches.
We first bound its exact rank, then determine the low-index growth rate,
and finally bound the largest Schmidt probability to obtain the high-index
regime. The von Neumann law in Eq.~\eqref{main-eq:op-vn} of the main
text was proved in Sec.~\ref{supp-subsec:binomial-vn}; the corresponding state problem is treated in
Sec.~\ref{supp-subsec:state} below.

\subsection{Count blocks and bounds on the exact rank}
\label{supp-subsec:blocks}

The operator specialization of the count-sector decomposition
\eqref{supp-eq:coherent-count-state} is
\begin{equation*}
 \ket{\Omega_t}
 =\sum_{k=0}^{t-1}\sum_{\ell=0}^{t}
 \sqrt{\pi^{\op}_{k\ell}}\ket{\widehat\phi_{k\ell}}.
\end{equation*}
As in Sec.~\ref{supp-subsec:binomial-vn}, the counts are independent, with
$K\sim\Bin(t-1,1/2)$ and $L\sim\Bin(t,1/2)$. Their joint weights are
\begin{equation}
 \pi^{\op}_{k\ell}
 =2^{-(2t-1)}\binom{t-1}{k}\binom{t}{\ell}.
 \label{supp-eq:count-weights}
\end{equation}

We now pass from this state-vector decomposition to the matrix language
needed for rank and Schatten-norm estimates. Choose orthonormal bases
$\{\ket{\mu}_{\rm L}\}$ and $\{\ket{\nu}_{\rm R}\}$ of the full left and
right output Hilbert spaces, respectively, with each basis vector lying
in a definite count sector. We define the normalized coefficient matrix
$M$ by
\begin{equation*}
 \ket{\Omega_t}=\sum_{\mu,\nu}M_{\mu\nu}
 \ket{\mu}_{\rm L}\ket{\nu}_{\rm R}.
\end{equation*}
This is the same matrix used in Sec.~\ref{supp-subsec:hermitian-source}. Its
singular values are the Schmidt coefficients of $\ket{\Omega_t}$, so
their squares are the normalized operator-Schmidt probabilities $p_j$.
For each count pair, let $M_{k\ell}$ be the block obtained by restricting
the rows and columns to the ranges of $P_k^{\rm L}$ and $P_\ell^{\rm R}$.
Equivalently, $M_{k\ell}$ is the coefficient matrix of the weighted
branch $\sqrt{\pi^{\op}_{k\ell}}\ket{\widehat\phi_{k\ell}}$. Hence
\begin{equation*}
 M=\sum_{k,\ell}M_{k\ell},\qquad
 \|M\|_{\HS}=1,\qquad
 \|M_{k\ell}\|_{\HS}^2=\pi^{\op}_{k\ell}.
\end{equation*}
For any matrix $A$, let $s_j(A)$, $1\le j\le\rank A$, denote its
nonzero singular values, counted with multiplicity. They are the
positive square roots of the nonzero eigenvalues of $A^\dagger A$.
The flat branch spectrum in Eq.~\eqref{supp-eq:op-core} therefore gives
\begin{equation}
 s_j(M_{k\ell})^2=\pi^{\op}_{k\ell}2^{-g(d)},\qquad
 1\le j\le2^{g(d)},\qquad
 d=\ell-k.
 \label{supp-eq:block-spectrum}
\end{equation}
These are spectra of individual blocks. Blocks sharing a row or column
count need not have orthogonal Schmidt vectors on that side, so we must
bound the spectrum of their sum.

At fixed time, the limit of vanishing R\'enyi index counts all nonzero
Schmidt probabilities equally. Writing $\OSR(t)$ for the operator Schmidt
rank, we obtain
\begin{equation*}
 S_0^\op(t):=\lim_{\alpha\to0^+}S_\alpha^\op(t)
 =\log\#\{j:p_j(t)>0\}=\log\OSR(t).
\end{equation*}
We determine its long-time growth directly from bounds on this rank.
For the lower bound, we apply the left and right count projectors to
isolate a branch with a large Schmidt rank. These projections cannot
increase Schmidt rank, so the selected branch gives a lower bound on the
full OSR. For the upper bound, rank is subadditive under addition. The
extreme block $(k,\ell)=(0,t)$ has rank $2^{t-1}$, giving
\begin{equation}
 2^{t-1}\le\OSR(t)=\rank M
 \le\sum_{k=0}^{t-1}\sum_{\ell=0}^t2^{g(\ell-k)}
 =2\sum_{j=1}^t j\,2^{t-j}=2^{t+2}-2t-4.
 \label{supp-eq:rank-bounds}
\end{equation}
Taking logarithms gives $S_0^\op(t)=t\log2+O(1)$. Thus the exact rank
already grows exponentially because an extreme count branch retains an
extensive number of Schmidt bits.

\subsection{R\'enyi indices \texorpdfstring{$0<\alpha<1$}{0 < alpha < 1}}
\label{supp-subsec:low}

For fixed $0<\alpha<1$, as $t\to\infty$, we compare the full
R\'enyi moment with the sum of the block moments, both defined below in
Eq.~\eqref{supp-eq:moments}. We show that their logarithms differ by at most
$O(\log t)$.

For $p>0$ and a matrix $A$ of rank $r$, we define
\begin{equation*}
 \|A\|_p^p=\sum_{j=1}^{r}s_j(A)^p
 =\operatorname{tr}\!\left[(A^\dagger A)^{p/2}\right].
\end{equation*}
For $p\ge1$, $\|A\|_p$ is the Schatten $p$-norm; for $0<p<1$, it is
the corresponding quasi-norm. Since $p_j=s_j(M)^2$,
Eq.~\eqref{supp-eq:block-spectrum} gives the full and block moments as
\begin{equation}
 \begin{aligned}
 Z_\alpha
 &=\sum_{j=1}^{\rank M}p_j^\alpha
 =\|M\|_{2\alpha}^{2\alpha},\\
 z_{k\ell}
 &=\sum_{j=1}^{2^{g(d)}}s_j(M_{k\ell})^{2\alpha}
 =(\pi^{\op}_{k\ell})^{\alpha}2^{(1-\alpha)g(d)}.
 \end{aligned}
 \label{supp-eq:moments}
\end{equation}
By definition, the R\'enyi entropy in the present range is
$S_\alpha^\op(t)=(1-\alpha)^{-1}\log Z_\alpha$.
There are $N_{\rm blocks}=t(t+1)$ blocks. Since
$M_{k\ell}=P_k^{\rm L}MP_\ell^{\rm R}$ and
$P_k^{\rm L},P_\ell^{\rm R}$ are orthogonal projectors, left and right
multiplication by them cannot increase any singular
value~\cite{supp-BhatiaMatrixAnalysis}. Thus $s_j(M_{k\ell})\le s_j(M)$ for
every $j$, and hence $z_{k\ell}\le Z_\alpha$ for every $\alpha>0$.

For the upper bound, set $p=2\alpha$ and regard every $M_{k\ell}$ as a
matrix on the full left and right Hilbert spaces. If $p\ge1$, the triangle
inequality followed by H\"older's inequality gives
\begin{equation*}
 \|M\|_p^p
 \le\left(\sum_{k,\ell}\|M_{k\ell}\|_p\right)^p
 \le N_{\rm blocks}^{p-1}
 \sum_{k,\ell}\|M_{k\ell}\|_p^p.
\end{equation*}
If $0<p<1$, the $p$-triangle inequality for the Schatten quasi-norm
power~\cite{supp-McCarthyCp1967} gives directly
\begin{equation*}
 \|M\|_p^p\le\sum_{k,\ell}\|M_{k\ell}\|_p^p.
\end{equation*}
Since $\|M_{k\ell}\|_p^p=z_{k\ell}$, the two cases combine into
\begin{equation}
 \max_{k,\ell}z_{k\ell}\le Z_\alpha
 \le N_{\rm blocks}^{\max(2\alpha-1,0)}\sum_{k,\ell}z_{k\ell}.
 \label{supp-eq:moment-bounds}
\end{equation}
Since the sum has only $N_{\rm blocks}$ terms, Eq.~\eqref{supp-eq:moment-bounds}
implies
\begin{equation*}
 \log Z_\alpha=\max_{k,\ell}\log z_{k\ell}+O(\log t).
\end{equation*}
By Eq.~\eqref{supp-eq:moments},
\begin{equation*}
 \log z_{k\ell}=\alpha\log\pi^{\op}_{k\ell}
 +(1-\alpha)g(\ell-k)\log2.
\end{equation*}
The second term rewards a large count imbalance, which produces many
Schmidt coefficients within the branch, whereas the first penalizes the
small probability of such an atypical branch. We therefore first determine
the largest branch probability at fixed extensive imbalance
$d=\ell-k=\eta t+O(1)$, and then optimize over $\eta$.

For the first step, write $k=(t-1)(1+u)/2$ and
$\ell=t(1+v)/2$, with $u,v\in[-1,1]$. At fixed $d$, these variables
obey the constraint $tv-(t-1)u=2d-1$. Stirling's formula then gives
\begin{equation}
 \log\pi^{\op}_{k\ell}=-(t-1)J(u)-tJ(v)+O(\log t),\qquad
 J(u)=\tfrac12\bigl[(1+u)\log(1+u)+(1-u)\log(1-u)\bigr],
 \label{supp-eq:stirling}
\end{equation}
where $0\log0=0$. Since $J$ is even and convex, the cost is minimized at
$-u=v=(2d-1)/(2t-1)$. Thus the least costly way to produce the
prescribed imbalance uses equal and opposite relative count deviations
in the two packets. The discreteness of $k$ and $\ell$ only affects the
$O(\log t)$ correction. Hence, for an admissible imbalance with
$|d|=\eta t+O(1)$,
\begin{equation}
 \max_{\ell-k=d}\log\pi^{\op}_{k\ell}
 =-tI(\eta)+O(\log t),\qquad
 I(\eta)=2J(\eta),\qquad 0\le\eta\le1.
 \label{supp-eq:rate}
\end{equation}
The finite source offset affects only the remainder. Since
$g(d)=|d|+O(1)$, substituting Eq.~\eqref{supp-eq:rate} into
Eq.~\eqref{supp-eq:moment-bounds} yields
\begin{equation}
 \log Z_\alpha
 =t\max_{0\le\eta\le1}
 \bigl[(1-\alpha)\eta\log2-\alpha I(\eta)\bigr]+O(\log t).
 \label{supp-eq:variational}
\end{equation}
The derivative $I'(\eta)=2\operatorname{artanh}\eta$ fixes the unique
maximum at $\eta_\alpha=\tanh x_\alpha$, where
$x_\alpha=(1-\alpha)\log2/(2\alpha)$. Using
$I(\tanh x)=2x\tanh x-2\log\cosh x$ and dividing by $1-\alpha$, we find
\begin{equation}
 S_\alpha^\op(t)=v_\alpha t+O(\log t),\qquad
 v_\alpha=\frac{2\alpha}{1-\alpha}\log\cosh x_\alpha.
 \label{supp-eq:velocity}
\end{equation}
This determines the low-index growth rate explicitly. Moreover,
$\lim_{\alpha\to0^+}v_\alpha=\log2$, in agreement with the growth of
$S_0^\op(t)$ established independently by Eq.~\eqref{supp-eq:rank-bounds}.

\subsection{R\'enyi indices \texorpdfstring{$\alpha>1$}{alpha > 1}}
\label{supp-subsec:high}

For $\alpha>1$, the largest Schmidt probability controls the R\'enyi
entropy from both sides:
\begin{equation*}
 -\log p_{\max}\le S_\alpha^\op(t)
 \le\frac{\alpha}{\alpha-1}(-\log p_{\max}).
\end{equation*}
The lower bound follows from R\'enyi monotonicity, while the upper bound
follows from $\sum_jp_j^\alpha\ge p_{\max}^\alpha$. It is therefore
sufficient to determine the scale of $p_{\max}=\|M\|_{\op}^2$, where
\begin{equation*}
 \|A\|_{\op}:=\sup_{\|v\|=1}\|Av\|=s_{\max}(A)
\end{equation*}
is the operator norm, equivalently the Schatten-$\infty$ norm.
For the factorization below, we write the two marginal binomial weights as
\begin{equation*}
 q_k=\Prob(K=k)=2^{-(t-1)}\binom{t-1}{k},\qquad
 r_\ell=\Prob(L=\ell)=2^{-t}\binom{t}{\ell},
\end{equation*}
so that $\pi^{\op}_{k\ell}=q_kr_\ell$.
The maxima of these marginal weights follow from the central-binomial
estimate
\begin{equation}
 \max_{0\le j\le n}2^{-n}\binom{n}{j}
 =\sqrt{\frac{2}{\pi n}}\left[1+O(n^{-1})\right],
 \qquad n\to\infty.
 \label{supp-eq:max-binomial-prob}
\end{equation}

We first inspect the count blocks separately. Their flat spectra in
Eq.~\eqref{supp-eq:block-spectrum} give
\begin{equation}
 \max_{k,\ell}\|M_{k\ell}\|_{\op}^2
 =\max_{k,\ell}q_kr_\ell\,2^{-g(\ell-k)}
 =(\max_k q_k)(\max_\ell r_\ell)
 \sim\frac{2}{\pi t}.
 \label{supp-eq:blockwise-pmax}
\end{equation}
For the second equality, $g(d)\ge0$ gives the upper bound by the product
of the two marginal maxima. This bound is attained at
$k=\ell=\lfloor t/2\rfloor$, where both marginal probabilities are
maximal and $g(0)=0$. The final asymptotic relation follows from
Eq.~\eqref{supp-eq:max-binomial-prob}, applied with $n=t-1$ for $q_k$ and
$n=t$ for $r_\ell$. Projecting $M$ onto a count block cannot
increase its operator norm, so Eq.~\eqref{supp-eq:blockwise-pmax} already proves
\begin{equation*}
 p_{\max}\ge\max_{k,\ell}\|M_{k\ell}\|_{\op}^2=\Theta(t^{-1}).
\end{equation*}
If the blocks were mutually orthogonal on both Schmidt sides, the
blockwise estimate would be the full answer. Blocks sharing a row or
column count can, however, combine coherently. We must therefore show
that this assembly does not produce a parametrically larger singular
value.

To bound $\|M\|_{\op}$ from above, we replace each count block by its
operator norm and form the nonnegative scalar matrix
\begin{equation}
 N_{k\ell}=\|M_{k\ell}\|_{\op}
 =\sqrt{q_k}\,2^{-g(\ell-k)/2}\sqrt{r_\ell}.
 \label{supp-eq:majorant}
\end{equation}
For a block vector $v=(v_\ell)$, the triangle inequality gives
$\|(Mv)_k\|\le\sum_\ell N_{k\ell}\|v_\ell\|$.
Taking the Euclidean norm over $k$ proves
$\|M\|_{\op}\le\|N\|_{\op}$.
Factor $N=\Lambda_K^{1/2}G\Lambda_L^{1/2}$, where
$\Lambda_K=\operatorname{diag}(q_k)$,
$\Lambda_L=\operatorname{diag}(r_\ell)$, and
$G_{k\ell}=2^{-g(\ell-k)/2}$. Every row and column sum of $G$ is bounded by
\begin{equation}
 C=\sum_{d\in\mathbb Z}2^{-g(d)/2}
 =\frac{2}{1-2^{-1/2}}.
 \label{supp-eq:kernel-sum}
\end{equation}
For any vector $x$, applying Cauchy--Schwarz to each row and then using
the column-sum bound gives $\|Gx\|^2\le C^2\|x\|^2$, and hence
$\|G\|_{\op}\le C$. Therefore
\begin{equation}
 p_{\max}\le C^2\bigl(\max_k q_k\bigr)
 \bigl(\max_\ell r_\ell\bigr)=O(t^{-1}).
 \label{supp-eq:pmax-upper}
\end{equation}
The product of marginal maxima in the last step was evaluated in
Eq.~\eqref{supp-eq:blockwise-pmax}.
Combined with the blockwise lower bound above, this proves
$p_{\max}=\Theta(t^{-1})$: nearly balanced branches determine the scale
of the largest Schmidt probability even after coherent assembly.
Substituting this result into the bounds at the beginning of this
subsection gives
\begin{equation}
 \log t+O(1)\le S_\alpha^\op(t)
 \le\frac{\alpha}{\alpha-1}\log t+O(1),\qquad 1<\alpha<\infty.
 \label{supp-eq:high-bounds}
\end{equation}
At $\alpha=\infty$, the entropy equals $-\log p_{\max}=\log t+O(1)$.
This proves the logarithmic regime, including its endpoint. The bounds
do not fix the logarithmic coefficient for finite $\alpha>1$.

\subsection{Domain-wall state}
\label{supp-subsec:state}

The same argument applies to the state problem with the weights
$\pi^{\dw}_{k\ell}=4^{-t}\binom tk\binom t\ell$ from
Eq.~\eqref{main-eq:dw-branches} of the main text, now with
$0\le k,\ell\le t$, and with $g(d)$ replaced by $|k-\ell|$.
Let $M^{\dw}$ be the normalized coefficient matrix of the domain-wall
state in orthonormal bases adapted to its output counts, analogous to
$M$ in Sec.~\ref{supp-subsec:blocks}. The extreme branch and the sum of
block ranks give
\begin{equation}
 2^t\le\rank M^{\dw}
 \le\sum_{k,\ell=0}^t2^{|k-\ell|}=2^{t+3}-3t-7,
 \label{supp-eq:state-rank}
\end{equation}
so $S_0^{\dw}=t\log2+O(1)$. For the low-index calculation, the
constraint on relative deviations becomes $(u-v)/2=\pm\eta$.
Since $J$ is even and convex, the minimum cost is again $2J(\eta)$,
giving the same rate $v_\alpha$. For high indices, the kernel
$2^{-|k-\ell|/2}$ has bounded row and column sums, and a balanced block
$k=\ell$ has rank one and weight $\Theta(t^{-1})$.
Thus the same upper and lower bounds on $p_{\max}$, and hence the same
high-index entropy bounds, follow. Together with the von Neumann result
proved in Sec.~\ref{supp-subsec:binomial-vn}, this
establishes the stated hierarchy for the domain-wall state.

\section{Lower bound on the bond dimension for MPO simulations}
\label{supp-sec:mpo}

The entanglement entropies alone do not always determine whether the
time-evolved operator can be simulated efficiently with tensor-network
methods. We therefore ask how large an MPO bond dimension is necessary to
retain a fixed fraction of the operator's Hilbert--Schmidt weight. We prove that for
every fixed error $0<\varepsilon<1$, the required bond dimension is at
least $\exp(c\sqrt t)$ for some $c=c(\varepsilon)>0$ and all sufficiently
large $t$. The proof has two steps. In
Sec.~\ref{supp-subsec:retained-weight} we use the flat spectra of the count
branches to bound the weight retained by any finite-rank approximation.
In Sec.~\ref{supp-subsec:fixed-accuracy} we combine this bound with the
$\sqrt t$ fluctuations of the count imbalance.

\subsection{A necessary-rank bound from the count branches}
\label{supp-subsec:retained-weight}

Order the full operator-Schmidt probabilities as
$p_1(t)\ge p_2(t)\ge\cdots$ and define the weight retained by the largest
$\chi$ probabilities as $F_\chi(t)=\sum_{j\le\chi}p_j(t)$. Fix
$0<\varepsilon<1$, where $\varepsilon$ is the relative squared
Hilbert--Schmidt error. For the normalized coefficient matrix $M$
introduced in Sec.~\ref{supp-subsec:blocks}, Schmidt truncation gives
\begin{equation}
 \min_{\rank X\le\chi}\|M-X\|_{\HS}^2=1-F_\chi(t),\qquad
 \chi_\varepsilon(t)=\min\{\chi:F_\chi(t)\ge1-\varepsilon\}.
 \label{supp-eq:approximation-error}
\end{equation}
Across the chosen cut, every MPO of bond dimension $\chi$ has a
coefficient matrix of rank at most $\chi$. A lower bound on
$\chi_\varepsilon(t)$ therefore gives the same necessary lower bound on
the bond dimension of every MPO approximation at this accuracy.

For the estimate below, we use the corresponding variational
characterization of the retained weight, furnished by the
Eckart--Young theorem~\cite{supp-EckartYoung1936}:
\begin{equation}
 F_\chi(t)=
 \max_{\substack{\rank Y\le\chi\\\|Y\|_{\HS}=1}}
 \left|\operatorname{tr}(Y^\dagger M)\right|^2.
 \label{supp-eq:weight-variational}
\end{equation}
Let $K\sim\Bin(t-1,1/2)$ and $L\sim\Bin(t,1/2)$ be independent, as in
Eq.~\eqref{supp-eq:count-weights}. Their joint probability
$\Prob(K=k,L=\ell)=\pi^{\op}_{k\ell}$ is the weight of the count branch
$(k,\ell)$. For a fixed branch $(k,\ell)$, the flat Schmidt rank is
$2^{g(\ell-k)}$, so its base-two logarithmic rank is $g(\ell-k)$. We
collect these branch-dependent values into the random variable
$G_t=g(L-K)$.

\begin{unnumberedlemma}[Retained-weight bound]
For every integer $\chi\ge1$, the retained weight satisfies
\begin{equation}
 F_\chi(t)\le\sum_{k=0}^{t-1}\sum_{\ell=0}^{t}
 \pi^{\op}_{k\ell}\min\{1,\chi2^{-g(\ell-k)}\}
 =\E\min\{1,\chi2^{-G_t}\}.
 \label{supp-eq:necessary-rank-bound}
\end{equation}
\end{unnumberedlemma}

\begin{proof}
Take a normalized trial matrix $Y$ of rank at most
$\chi$ in Eq.~\eqref{supp-eq:weight-variational}. Obtain its count block
$Y_{k\ell}$ using the same row and column projections as for
$M_{k\ell}$. This block also has rank at most $\chi$, so it can capture
at most a fraction $\min\{1,\chi2^{-g(\ell-k)}\}$ of the flat branch
spectrum. Therefore
\begin{equation}
 \left|\operatorname{tr}(Y_{k\ell}^\dagger M_{k\ell})\right|
 \le\|Y_{k\ell}\|_{\HS}
 \sqrt{\pi^{\op}_{k\ell}\min\{1,\chi2^{-g(\ell-k)}\}}.
 \label{supp-eq:block-overlap}
\end{equation}
The block decomposition gives
\begin{equation*}
 \operatorname{tr}(Y^\dagger M)
 =\sum_{k,\ell}\operatorname{tr}(Y_{k\ell}^\dagger M_{k\ell}).
\end{equation*}
The count blocks are orthogonal in the Hilbert--Schmidt inner product,
and $\sum_{k,\ell}\|Y_{k\ell}\|_{\HS}^2\le1$. Summing
Eq.~\eqref{supp-eq:block-overlap} and applying the Cauchy--Schwarz inequality
proves
Eq.~\eqref{supp-eq:necessary-rank-bound}.
\end{proof}

\subsection{Fixed-accuracy lower bound}
\label{supp-subsec:fixed-accuracy}

The physical origin of the large necessary rank is the typical count
imbalance. Since $L-K+t-1\sim\Bin(2t-1,1/2)$, the central limit theorem
implies that $G_t/\sqrt t$ approaches a continuous half-Gaussian
distribution. Thus a typical branch carries a number of Schmidt bits of
order $\sqrt t$ and has rank exponential in $\sqrt t$. The continuity of
the limiting distribution at the origin means that, for every fixed
$0<\varepsilon<1$, we can choose $c_\varepsilon>0$ small enough that the
limiting probability of
$G_t\le 2c_\varepsilon\sqrt t/\log2$ is strictly smaller than
$1-\varepsilon$.

We now set $\chi=\lfloor\exp(c_\varepsilon\sqrt t)\rfloor$ in
Eq.~\eqref{supp-eq:necessary-rank-bound} and separate the branches according to
the size of their imbalance $G_t$. The branches satisfying
$G_t\le 2c_\varepsilon\sqrt t/\log2$ contribute at most their total
probability, because the factor inside the average in
Eq.~\eqref{supp-eq:necessary-rank-bound} is bounded by one. For every remaining
branch,
\begin{equation*}
 \chi2^{-G_t}
 \le \exp(c_\varepsilon\sqrt t)\exp(-2c_\varepsilon\sqrt t)
 =\exp(-c_\varepsilon\sqrt t).
\end{equation*}
Since the total probability of the remaining branches is at most one,
Eq.~\eqref{supp-eq:necessary-rank-bound} gives, for all sufficiently large $t$,
\begin{equation}
 F_{\lfloor\exp(c_\varepsilon\sqrt t)\rfloor}(t)
 \le \Prob\!\left(
 G_t\le\frac{2c_\varepsilon}{\log2}\sqrt t\right)
 +\exp(-c_\varepsilon\sqrt t)
 <1-\varepsilon.
 \label{supp-eq:fixed-rank-failure}
\end{equation}
This rank therefore fails the accuracy condition in
Eq.~\eqref{supp-eq:approximation-error}, which proves
\begin{equation}
 \chi_\varepsilon(t)\ge\exp(c_\varepsilon\sqrt t)
 \label{supp-eq:fixed-error}
\end{equation}
for all sufficiently large $t$.

Equation~\eqref{supp-eq:necessary-rank-bound} also implies that every
approximation whose rank grows polynomially has asymptotically vanishing
overlap. Indeed, for $\chi(t)=O(t^p)$ with fixed $p$, splitting the bound at
$G_t=t^{1/4}$ makes both terms vanish: the low-imbalance probability
vanishes by the same central-limit argument, while
$\chi(t)2^{-t^{1/4}}\to0$. Hence every polynomial rank retains an
asymptotically vanishing fraction of the normalized Hilbert--Schmidt
weight.

\section{Conjectured scaling for a reflection-invariant circuit}
\label{supp-sec:reflection}

We now consider whether the same entanglement-growth laws can occur
with spatial reflection symmetry. For this purpose, we introduce a new
circuit with local dimension $D=8$ and specify its gate. The local space is
$H_{\rm sec}\otimes H_{\rm col}$, with $H_{\rm sec}\simeq\mathbb C^2$ and
$H_{\rm col}\simeq\mathbb C^4$. The sector labels are $A,B$, while the color
basis is labeled by the elements of the field $\mathbb F_4$. Let
$\omega\in\mathbb F_4$ satisfy $\omega^2=\omega+1$.
The gate is defined as follows. It acts as the identity when the sector
labels coincide. For different sector labels, its classical update rule is
\begin{equation}
 \begin{split}
  r(A,a;B,b)&=(B,a+\omega b;A,\omega a+\omega b),\\
  r(B,a;A,b)&=(A,\omega a+\omega b;B,\omega a+b),
 \end{split}
 \label{supp-eq:reflection-gate}
\end{equation}
where $a,b\in\mathbb F_4$ and all arithmetic is in this field.
The gate is involutive and satisfies the braid relation.

Based on preliminary computations, we conjecture that this
reflection-invariant circuit exhibits the same asymptotic
entanglement-growth scaling forms as the $D=4$ model, with only the
prefactors differing.

%
\end{document}